\documentclass[a4paper,11pt]{article}
\usepackage[utf8]{inputenc}
\usepackage[english]{babel}
\usepackage[margin=1in]{geometry}
\usepackage{amsmath}
\usepackage{amsfonts}
\usepackage{amssymb}
\usepackage{graphicx}
\usepackage{tikz}
\usepackage{booktabs}
\usepackage{pgfpages}
\usepackage{enumitem}
\usepackage{subfigure}
\usepackage{todonotes}

\usepackage{amsthm}

\newtheoremstyle{thm}	
  {3pt}			
  {3pt}			
  {\itshape}		
  {0pt}			
  {\bfseries}		
  {.}			
  {6pt}		
  {}			

\numberwithin{equation}{section}  
\newtheorem{defn}{Definition}[section]
\newtheorem{exa}[defn]{Example}
\newtheorem{rem}[defn]{Remark}

\newtheorem{prop}[defn]{Proposition}

\newtheorem{cor}[defn]{Corollary}
\newtheorem{lem}[defn]{Lemma}
\newtheorem{asu}[defn]{Assumption}

\newcommand{\D}{{\mathop{}\!\mathrm{d}}} 
\newcommand{\R}{\mathbb{R}}
\newcommand{\Q}{\mathbb{Q}}
\newcommand{\C}{\mathcal{C}}
\newcommand{\Clin}{\mathcal{C}_{\operatorname{lin}}}
\newcommand{\N}{\mathbb{N}}

\newcommand{\PP}{\mathbb{P}}

\newcommand{\E}{\mathbb{E}}

\newcommand{\one}{ 1 \hspace{-3pt} \mathrm{l}} 

\usepackage{hyperref}
\usepackage{xcolor}
\hypersetup{
    colorlinks,
    linkcolor={red!50!black},
    citecolor={blue!50!black},
    urlcolor={blue!80!black}
}

\author{Julian Sester}

\title{On intermediate Marginals in Martingale Optimal Transportation}
\date{
\small\textit{National University of Singapore, Department of Mathematics,\\ 21 Lower Kent Ridge Road, 119077.}                                                                                                                             \\ \vspace{0.5cm}
\normalsize \today }
\begin{document}
\maketitle

\begin{abstract}
We study the influence of additional intermediate marginal distributions on the value of the martingale optimal transport problem. From a financial point of view, this corresponds to taking into account call option prices not only, as usual, for those call options where the respective future maturities coincide with the maturities of some exotic derivative but also additional maturities and then to study the effect on model-independent price bounds for the exotic derivative. We characterize market settings, i.e., combinations of the payoff of exotic derivatives, call option prices and marginal distributions that guarantee improved price bounds as well as those market settings that exclude any improvement.

 Eventually, we showcase in numerous examples that the consideration of additional price information on vanilla options may have a considerable impact on the resultant model-independent price bounds.
\end{abstract}

\begin{center}
\end{center}

\section{Introduction}
We consider a single underlying security $S=(S_t)_{t\geq 0}$ and aim at determining model-free price bounds for a financial derivative with an associated measurable payoff function $c:\R^n \rightarrow \R $ that depends on the value $S_{t_i}$ at future maturities $t_i$ with $i=1,\dots,n$. Then, relying on the martingale optimal transport (MOT) approach introduced in \cite{beiglbock2013model}, \cite{beiglboeck2016problem} and \cite{henry2016explicit}, we consider for the valuation of the exotic derivative $c$ all martingale models calibrated to prices of vanilla options (for all strikes) written on $S$ that expire at times $(t_i)_{i=1,\dots,n}$. In this setting, solving MOT problems corresponds to the calculation of the most extreme prices for $c$ among all arbitrage-free models that are calibrated to the observed vanilla option prices. 

The crucial observation which motivates the present paper is that in practice, due to the broad range of offered maturities for traded vanilla options\footnote{Note that in practice the expiration date of an exchange-traded call option is often the third Friday of a month. Very often, for shorter maturities more expiration dates are available, often on a weekly basis.}, it is often possible to calibrate a candidate underlying model to more maturities than those dates on which the payoff of the derivative $c$ depends. Consequently, to determine model-free price bounds one may additionally take into account price information on vanilla options maturing at some times $(s_j)_{j \in J}$ such that $s_j \neq t_i$ for all $i=1,\dots,n$ and for all $j\in J$, where $J$ is some index set associated to the additionally considered maturities.

 Mathematically, when computing price bounds via MOT, this corresponds to restricting the set of admissible joint distributions to those distributions that are consistent with auxiliary one-dimensional marginal distributions (that are implied by call option prices due to a result from \cite{breeden1978prices}) as well as to introduce additional intermediate martingale constraints. This means we restrict the dynamics of admissible underlying processes at more future dates.     

In this paper we pursue the question if this additional calibration procedure leads to restrictions which have an observable influence on the model-independent price bounds of exotic derivative{\color{blue}s}.

Framing an answer to this question is not trivial. Already in the two-marginal case, i.e., $n=2$, it is possible to derive a simple example where the value of the martingale transport problem is influenced by restrictions imposed through an additional third intermediate marginal distribution, compare Example~\ref{exa_theoretical_1}, but it is also possible to construct settings in which the price is not influenced by additional intermediate marginals, see Section~\ref{sec_convex_interpol} and \ref{sec_linear_interpol}.


To study the above outlined practice-implied question, we employ the MOT-approach and its dual formulation as discussed and studied in \cite{backhoff2022stability}, \cite{bayraktar2022supermartingale}, \cite{beiglbock2013model}, \cite{beiglbock2022potential}, \cite{beiglbock2019dual}, \cite{beiglbock2017complete},  \cite{cheridito2021martingale}, \cite{cox2019structure}, \cite{doldi2020entropy}, \cite{ekren2018constrained}, \cite{henry2016explicit},  \cite{henry2016explicit_full}, \cite{hobson2015robust}, \cite{hobson2019left}, \cite{hobson2019robust}, \cite{juillet2016stability}, \cite{wiesel2019continuity} among many others.

In a recent work, \cite{nutz2017multiperiod} study a related problem which can be regarded as the inverse problem of the problem treated in this paper: let $c:\R^n \rightarrow \R$ be an arbitrary measurable payoff function of a financial derivative. Then, the authors address the question how the original MOT-problem is influenced and has to be redefined when only $\mu_1$ and $\mu_n$ are fixed and one has no information on intermediate marginals $\mu_2,\dots,\mu_{n-1}$. Moreover, the payoff function $c$ is allowed to depend on all of the associated times $t_1,t_2,\dots,t_n$, i.e., the authors from \cite{nutz2017multiperiod} study the minimization problem
\[
\inf_{\Q \in \mathcal{M}^{n}(\mu_1,\mu_n)}\E_\Q[c(S_{t_1},\dots,S_{t_n})]
\]
where $\mathcal{M}^{n}(\mu_1,\mu_n)$ denotes the set of martingale measures on $\R^n$ with prescribed first and last marginals $\mu_1,\mu_n$. In contrast, we want to study the minimization problem
\[
\inf_{\Q \in \mathcal{M}(\mu_1,\mu_2,\cdots,\mu_{n-1},\mu_n)}\E_\Q[\tilde{c}(S_{t_1},S_{t_n})]
\]
for some $\tilde{c}:\R^2 \rightarrow \R$.
 In a remark within the book \cite{henry2017modelfree} another version of the problem is discussed. More specifically, in \cite{henry2017modelfree}, the author investigates the problem when in the dual variant of the problem, established in \cite{beiglbock2013model}, trading in European options at intermediate points is prohibited but trading in the underlying security is allowed. We will discuss the result from \cite{henry2017modelfree} within this article, see Lemma~\ref{lem_intermediate_trading}. The authors from \cite{eckstein2021martingale} compare within a sole example the influence of additional marginal information on price bounds with the influence of this information in combination with the homogeneity of an underlying process. Further,~in the recent works \cite{aksamit2020robust}, \cite{ansari2022improved}, \cite{de2020bounds}, \cite{eckstein2021martingale},~\cite{lutkebohmert2019tightening}, \cite{lux2017improved}, \cite{neufeld2022model}, \cite{sester2020robust}, and \cite{tankov2011improved} improvements of robust and model-free price bounds either via the inclusion of prices of other derivatives, market information or via additional assumptions on the underlying process are studied.
  However, to the best of our knowledge, to explicitly study the influence of  intermediate marginals  on model-free price bounds is a novelty. 
  
From our point of view to study the price bound improvement induced by the inclusion of intermediate marginals is of particular importance, as an improvement obtained in this way does not require to expensively acquire information on other types of possibly OTC-traded financial instruments nor to impose additional assumptions on the joint distribution of the underlying process. The approach simply respects all information available to someone who applies the MOT-approach for the valuation of financial derivatives.
 
The contribution of this article is twofold. First, we give a mathematical characterization of the cases in which improved price bounds are guaranteed, and second, we study in several examples, also involving real financial data, the degree of improvement through the inclusion of intermediate marginal information.
 
The remainder of the paper is structured as follows. In Section \ref{sec_mot} we introduce the setting and present our main results that comprise a characterization of cases in which including additional intermediate marginals causes improved price bounds as well as a description of the degree of improvement through dual trading strategies.

In Section~\ref{sec_financial_application} we move away from the idealized setting in which we know the entire marginal distributions and discuss the influence of a finite amount of call option prices at intermediate times on model-free price bounds. 

In Section~\ref{sec_examples} we provide several examples also involving real market data. The proofs of the mathematical statements  can be found in Section~\ref{sec_proofs}.

\section{Setting and Main Results} \label{sec_mot}
In Section~\ref{sec_setting} we present the underlying setting of the paper that is used to derive a characterization for improved price bounds in Section~\ref{sec_main_result}. In Section~\ref{sec_degree_improvement} we discuss the degree of improvement emerging from considering an additional intermediate marginal. Section~\ref{sec_financial_application} provides a discussion of the case in financial markets where only a finite amount of traded call options are taken into account instead of the full marginal information.
\subsection{Setting}\label{sec_setting}
We introduce the setting that we use in this paper. We also refer to \cite[Section 1.1.]{beiglbock2013model} where a similar setting is considered. 

Let $\mathcal{P}(\R^d)$ denote the set of probability measures on $\R^d,d \in \N$ equipped with the Borel $\sigma$-algebra $\mathcal{B}(\R^d)$. Then, for $n \in \N$, we fix marginal distributions $\mu_1,\dots,\mu_n$ with $\mu_i \in \mathcal{P}(\R)$ for $i=1,\dots,n$ and we impose the following assumption ensuring the existence of martingale measures possessing these marginals, see also \cite{kellerer1972markov} and \cite{strassen1965existence}.
\begin{asu}\label{asu_marginals}
The marginals $(\mu_i)_{i=1,\dots,n}$ are assumed to increase in the convex order\footnote{A probability measure $\mu_2\in \mathcal{P}(\R)$ is larger in convex order than $\mu_1\in \mathcal{P}(\R)$, abbreviated by $\mu_1 \preceq \mu_2$, if $\int_\R f(x) \D \mu_1(x) \leq \int_\R f(x) \D \mu_2(x)$ for all convex functions $f:\R \rightarrow \R$ such that the integrals are finite.\\ Moreover, note that if the marginals increase in convex order, then the prices of European call options computed as expectations with respect to these marginals also increase, compare  \cite[Lemma 7.24]{follmer2016stochastic}. That prices of European call options increase with an increasing maturity is well-known in arbitrage free markets, hence the assumption appears to be natural.} $\preceq$. Moreover, the marginals $(\mu_i)_{i=1,\dots,n}$ are assumed to possess finite first moments.
\end{asu}
We denote by $\mathcal{M}(\mu_1,\dots,\mu_n) \subset \mathcal{P}(\R^n)$ the set of all $n$-dimensional martingale measures with fixed one-dimensional marginals $(\mu_i)_{i=1,\dots,n}$, which equivalently can be written as
\begin{align*}
\mathcal{M}(\mu_1,\dots,\mu_n):=\bigg\{\Q \in \mathcal{P}(\R^n)~\bigg|~&\int f(x_i)  \Q(\D x_1,\dots,\D x_n)=\int f(x_i)  \mu_i( \D x_i)~\\
&\int \Delta(x_1,\dots,x_j)(x_{j+1}-x_j) \D  \Q(\D x_1,\dots,\D x_n)=0 \\
&\text{ for all } \Delta \in \mathcal{C}_b(\R^j), f \in L^1(\mu_i),~j = 1,\dots,n-1,~i =1,\dots,n\bigg\},
\end{align*}
where $\mathcal{C}_b(\R^j)$ denotes the set of continuous and bounded functions $f:\R^j \rightarrow \R$ for $j \in \N$, compare also \cite[Lemma 2.3]{beiglbock2013model} for a characterization of $\mathcal{M}(\mu_1,\dots,\mu_n)$.

Following \cite{breeden1978prices}, marginal distributions of an underlying security under pricing measures can be inferred from observed market prices of vanilla call and put options.
Hence, for some measurable function $c:\R^n \rightarrow \R$ that can be interpreted as the payoff of a financial derivative depending on the underlying security at future times $t_1,\dots,t_n$, model-independent price bounds for this derivative can, in absence of interest rates, dividend yields and market frictions, be computed as 
\[
\inf_{\Q \in \mathcal{M}(\mu_1,\dots\mu_n)}\E_\Q[c]\qquad \text{ and } \qquad\sup_{\Q \in \mathcal{M}(\mu_1,\dots\mu_n)}\E_\Q[c].
\]
Our goal is now to study whether including information about additional intermediate marginal distributions will change these model-independent price bounds and whether it is possible to derive conditions that guarantee an improvement.
To reduce the notational complexity, and since the theoretical results turn out to be analogue in higher dimensions (compare Remark~\ref{rem_nmarginals}), we mainly discuss the problem in its easiest and clearest formulation which is present in the three-marginal case and if the payoff function $c$ depends only on the first and third of the three prospective times to which the marginal distributions are associated. 

To this end, let us assume\footnote{This assumption is mainly imposed to be able to apply a duality result (see \cite{beiglbock2013model} or \cite{bartl2019robust}), and could be relaxed to a certain degree, for example by considering semi-continuous payoff functions instead of continuous functions. To avoid the need of distinguishing between lower-semi continuous payoff functions for lower bounds and upper semi-continuous payoff functions for upper bounds, we decided to circumvent this issue by only considering continuous payoff functions. The linear growth condition could be slightly relaxed to functions that do not grow stronger than a sum of $\mu_i$-integrable functions, however it is standard to use the linear growth condition as in the definition of $\mathcal{C}_{\operatorname{lin}}(\R^d)$, see e.g. \cite[Theorem 1.1]{beiglbock2013model}.}, that the payoff function $c$ is continuous and of linear growth, i.e., we assume $c\in \mathcal{C}_{\operatorname{lin}}(\R^2)$, where
\[
\mathcal{C}_{\operatorname{lin}}(\R^d):=\left\{c\in \mathcal{C}(\R^d,\R)~\middle|~\sup_{(x_1,\dots,x_d)\in \R^d} \frac{|c(x_1,\dots,x_d)|}{1+\sum_{i=1}^d|x_i|} < \infty \text{ for all } (x_1,\dots,x_d)\in \R^d\right\} \text{ for } d \in \N.
\]

Then, we address the question for which choices of measurable cost functions $c\in \mathcal{C}_{\operatorname{lin}}(\R^2)$  and for which choices of marginal distributions $\mu_1,\mu_2,\mu_3 \in \mathcal{P}(\R)$ with $\mu_1 \preceq \mu_2 \preceq \mu_3$ it is true that
\begin{equation}\label{eq_equality of problems}
\inf_{\Q \in \mathcal{M}(\mu_1,\mu_2,\mu_3)}\E_\Q[c(S_{t_1},S_{t_3})]=\inf_{\Q \in \mathcal{M}(\mu_1,\mu_3)}\E_\Q[c(S_{t_1},S_{t_3})],
\end{equation}
or in the upper bound formulation
\begin{equation}\label{eq_equality of problems_sup}
\sup_{\Q \in \mathcal{M}(\mu_1,\mu_2,\mu_3)}\E_\Q[c(S_{t_1},S_{t_3})]=\sup_{\Q \in \mathcal{M}(\mu_1,\mu_3)}\E_\Q[c(S_{t_1},S_{t_3})],
\end{equation}
where $S=(S_{t_i})_{i=1,2,3}$ denotes the canonical process on $\R^3$. Equivalently, we study the cases when \eqref{eq_equality of problems} and \eqref{eq_equality of problems_sup} are violated, and therefore improved price bounds can be obtained through the inclusion of additional marginal information. For completeness, details of the general $n$-dimensional versions of the problems \eqref{eq_equality of problems} and \eqref{eq_equality of problems_sup} are discussed in Remark~\ref{rem_n_marginals}.
A financial interpretation of the influence of more intermediate marginals on price bounds corresponds, according to the rationale from \cite{breeden1978prices}, to the question whether incorporating more market-implied price information on liquidly traded vanilla options may influence the model-independent valuation problem for exotic derivatives with payoff $c$, even if the payoff of the derivative does not depend on the value of the underlying security at maturity of the additionally considered vanilla options.

This effect may occur since through the additional calibration to more marginal distributions, the set of possible arbitrage-free market models is reduced by excluding models which are not consistent with the additional marginal information. Moreover, additional martingale constraints taking into account those intermediate times are introduced. Thus, it may happen that  models which lead to extreme prices without considering the intermediate marginal are no more contained in the set of admissible models when taking into account a third marginal.
To study the influence of intermediate marginals on price bounds we are therefore interested in identifying choices of marginal distributions for which the non-empty and compact set\footnote{Compare e.g. \cite[Proposition 2.2.]{neufeld2021stability} for a proof of the compactness with respect to the Wasserstein distance and the non-emptiness, which both are consequences of Berge's maximum theorem.} set of optimal measures for the two-marginal MOT-problem
\[
\mathcal{Q}_c^*(\mu_1,\mu_3):= \left\{\Q^* \in \mathcal{M}(\mu_1,\mu_3)~\middle|~\inf_{\Q\in \mathcal{M}(\mu_1,\mu_3)}\E_\Q[c]=\E_{\Q^*}[c]\right\}
\]
is still \emph{consistent}  with the newly introduced intermediate marginal $\mu_2$. More precisely, consistency means in this context there exists some martingale measure $\Q \in \mathcal{M}(\mu_1,\mu_2,\mu_3)$ such that $\pi(\Q)\in  \mathcal{Q}_c^*(\mu_1,\mu_3)$ for $\pi$ being defined as the projection
\begin{align*}
\pi:\mathcal{M}(\mu_1,\mu_2,\mu_3)\rightarrow \mathcal{M}(\mu_1,\mu_3),~ \pi(\Q)(A_1,A_3):= \Q(A_1,\R,A_3),~\text{ for all }A_1,A_3 \in \mathcal{B}(\R),
\end{align*}
where $\mathcal{B}(\R)$ denotes the Borel-sets on $\R$.

The pricing-hedging duality result for MOT (see \cite{beiglbock2013model},  \cite{beiglbock2017complete}, \cite{cheridito2021martingale}) enables to regard the incorporation of additional intermediate marginals from another point of view: including more intermediate marginals allows for more flexible compositions of the semi-static trading strategies that can be considered for sub-hedging, as one can additionally trade statically in options expiring at an intermediate maturity and readjust the dynamic position in the underlying security at an associated intermediate time. This additional flexibility may allow to increase the maximal sub-replication price. We are therefore interested in studying the maximal improvement that is possible by switching from a two-marginal semi-static trading strategy to a three-marginal semi-static trading strategy\footnote{We refer for more details on semi-static trading strategies and its properties to \cite{acciaio2017space}, \cite{burzoni2017model}, \cite{davis2007range} and \cite{nutz2022limits}.}.

To this end, we define for $u_1 \in L^1(\mu_1)$, $u_3 \in L^1(\mu_3)$, $\Delta_1 \in \mathcal{C}_b(\R)$, $\Delta_2 \in \mathcal{C}_b(\R^2)$ and  $c\in \Clin(\R^2)$ the function
\begin{equation}\label{eq_defn_H}
\begin{aligned}
\R \ni x_2 \mapsto H_{(u_i),(\Delta_i)}(x_2):=\inf_{x_1,x_3 \in \R}\bigg\{ &c(x_1,x_3)-u_1(x_1)-u_3(x_3)\\
&-\Delta_1(x_1)(x_2-x_1)-\Delta_2(x_1,x_2)(x_3-x_2)\bigg\} 
\end{aligned}
\end{equation}
that describes the \emph{pointwise gap} between the payoff of $c$ and a model-free sub-hedging strategy with a payoff of the form $u_1(x_1)+u_3(x_3)+\Delta_1(x_1)(x_2-x_1)+\Delta_2(x_1,x_2)(x_3-x_2)$, i.e., with static trading in options expiring at times $t_1$, $t_3$ and with dynamic trading in the underlying security at times $t_1$, $t_2$, and $t_3$. An immediate consequence of the definition of $H_{(u_i),(\Delta_i)}$ is the following Lemma.
\begin{lem}\label{lem_positivity_of_H}
Let $u_1 \in L^1(\mu_1)$, $u_3 \in L^1(\mu_3)$, $\Delta_1 \in \mathcal{C}_b(\R)$, $\Delta_2 \in \mathcal{C}_b(\R^2)$ and $c\in \Clin(\R^2)$.
\begin{itemize}
\item[(i)]
We have that $H_{(u_i),(\Delta_i)}(x_2) \geq 0$ for all $x_2 \in \R$ is equivalent to 
\[
u_1(x_1)+u_3(x_3)+\Delta_1(x_1)(x_2-x_1)+\Delta_2(x_1,x_2)(x_3-x_2) \leq c(x_1,x_3)
\]
for all $x_1,x_2,x_3 \in \R$.
\item[(ii)]
We have 
\[
u_1(x_1)+H_{(u_i),(\Delta_i)}(x_2)+u_3(x_3)+\Delta_1(x_1)(x_2-x_1)+\Delta_2(x_1,x_2)(x_3-x_2) \leq c(x_1,x_3)
\]
for all $x_1,x_2,x_3 \in \R$.
\end{itemize}
\end{lem}
This means, by Lemma~\ref{lem_positivity_of_H}~(i), if $H_{(u_i),(\Delta_i)}(x_2) \geq 0$ for all $x_2 \in \R$, then the considered semi-static trading strategy sub-replicates $c$ pointwise. Moreover, by Lemma~\ref{lem_positivity_of_H}~(ii), adding a static trading position with payoff $H_{(u_i),(\Delta_i)}(S_{t_2})$ to the semi-static strategy $(u_i),(\Delta_j)$ still sub-replicates the derivative $c$. Since this observation remains true for those sub-replication strategies that lead to a maximal price for the two-marginal sub-hedging problem, we are interested to study whether an integration of $H_{(u_i),(\Delta_i)}$ with respect to the intermediate marginal $\mu_2$ leads to a higher price, and therefore to improved price bounds when considering optimal sub-hedging strategies. Note that in the following we often abbreviate pointwise inequalities of the form $f(x) \geq g(x)$ for all $x \in \R^m$, for some $m \in \N$, by writing $f \geq g$.

\subsection{A Characterization of Improved Price Bounds}\label{sec_main_result}
Below we characterize the equality described in \eqref{eq_equality of problems}. As we will point out in Remark~\ref{rem_upper_bound} and \ref{rem_nmarginals}, similar assertions can also be derived for the upper bound problem and for the case with $n$ marginals, where $n >2$. Note that the subsequent proposition allows to identify the cases in which improved price bounds through the inclusion of intermediate marginals can be excluded. Thus, simultaneously the cases where improvement can be guaranteed are characterized.

\begin{prop}\label{prop_intermediate_projection}
Let $c \in \mathcal{C}_{\operatorname{lin}}(\R^2)$ and assume that Assumption~\ref{asu_marginals} is fulfilled for the marginal distributions $\mu_1,\mu_2,\mu_3 \in \mathcal{P}(\R)$. Then, the following statements are equivalent.
\begin{itemize}
\item[(i)] We have
\[
\inf_{\Q \in \mathcal{M}(\mu_1,\mu_2,\mu_3)}\E_\Q[c(S_{t_1},S_{t_3})]=\inf_{\Q \in \mathcal{M}(\mu_1,\mu_3)}\E_\Q[c(S_{t_1},S_{t_3})].
\]
\item[(ii)]
There exists $u_1,u_3 \in \C_b(\R)$, $\Delta_1 \in \mathcal{C}_b(\R)$, $\Delta_2 \in \mathcal{C}_b(\R^2)$ with $H_{(u_i),(\Delta_i)}\geq 0$ such that for all $v_1,v_3 \in C_b(\R)$ and all $\widetilde{\Delta}_1 \in \mathcal{C}_b(\R)$, $\widetilde{\Delta}_2 \in \mathcal{C}_b(\R^2)$ we have
\[
\E_{\mu_2} \left[H_{(u_i-v_i),(\Delta_i-\widetilde{\Delta}_i)}\right]- \E_{\mu_1}[v_1]- \E_{\mu_3}[v_3] \leq   0.
\]
\item[(iii)]
For all $\varepsilon>0$ there exist $u_1 \in \mathcal{C}_b(\R),u_3 \in \mathcal{C}_b(\R)$, $\Delta_1 \in \mathcal{C}_b(\R)$ and $\Delta_2 \in \mathcal{C}_b(\R^2)$ with 
\begin{equation}\label{eq_defn_H}
\begin{aligned}
H_{(u_i),(\Delta_i)}(x_2) \geq 0 \text{ for all }x_2 \in \R,
\end{aligned}
\end{equation}
such that
\begin{equation}\label{eq_ineq_requirement_H_smaller_eps}
\E_{\mu_2}[H_{(u_i),(\Delta_i)}] < \varepsilon,
\end{equation}
and such that 
\begin{equation}\label{eq_ineq_requirement_lemma_ii}
\left|\E_{\mu_1}[u_1]+\E_{\mu_2}[H_{(u_i),(\Delta_i)}]+\E_{\mu_3}[u_3]-\inf_{\Q \in \mathcal{M}(\mu_1,\mu_2,\mu_3)}\E_\Q[c(S_{t_1},S_{t_3})]\right|<\varepsilon.
\end{equation}
\item[(iv)]
For all $\varepsilon>0$ there exist $u_1,u_3 \in \mathcal{C}_b(\R)$, $\Delta_1 \in \mathcal{C}_b(\R)$ such that
\begin{equation}\label{eq_ineq_requirement_lemma_iii}
u_1(x_1)+u_3(x_3)+\Delta_1(x_1)(x_3-x_1) \leq c(x_1,x_3) \text{ for all } x_1,x_3 \in \R,
\end{equation}
and such that 
\begin{equation}\label{eq_ineq_requirement_lemma_iii_eps}
\left|\E_{\mu_1}[u_1]+\E_{\mu_3}[u_3]-\inf_{\Q \in \mathcal{M}(\mu_1,\mu_2,\mu_3)}\E_\Q[c(S_{t_1},S_{t_3})]\right|<\varepsilon.
\end{equation}
\item[(v)]
For all $\varepsilon>0$ and for all $u_1 , u_2, u_3 \in \C_b(\R)$, $\Delta_1 \in \mathcal{C}_b(\R)$, $\Delta_2 \in \mathcal{C}_b(\R^2)$ with 
\begin{equation}\label{eq_ineq_3marginals}
\sum_{i=1}^3u_i(x_i)+{\Delta_1}(x_1)(x_2-x_1)+{\Delta_2}(x_1,x_2)(x_3-x_2) \leq c(x_1,x_3) \text{ for all } x_1,x_2,x_3 \in \R,
\end{equation}
there exist $v_1,v_3 \in \C_b(\R)$, $\widetilde{\Delta}_1 \in \mathcal{C}_b(\R)$ with
\begin{equation}\label{eq_ineq_2marginals}
v_1(x_1)+v_3(x_3)+\widetilde{\Delta}_1(x_1)(x_3-x_1) \leq c(x_1,x_3) \text{ for all } x_1,x_3 \in \R,
\end{equation}
such that 
\[
\left|\E_{\mu_1}[v_1]+\E_{\mu_3}[v_3]-\left(\sum_{i=1}^3\E_{\mu_i}[u_i]\right)\right|<\varepsilon.
\]
\item[(vi)] There exists some $\Q \in \mathcal{M}(\mu_1,\mu_2,\mu_3)$ such that $\pi(\Q) \in \mathcal{Q}_c^*(\mu_1,\mu_3)$.
\item[(vii)] There exists some $\Q \in \mathcal{M}(\mu_1,\mu_2,\mu_3)$ such that 
\[
\E_\Q[c(S_{t_1},S_{t_3})]=\inf_{\Q \in \mathcal{M}(\mu_1,\mu_3)}\E_\Q[c(S_{t_1},S_{t_3})].
\]
\item[(viii)]
There exist {probability} kernels\footnote{A {probability} kernel is a map where for fixed first component the map is a probability measure, and for fixed second component the map is Borel-measurable.} $\Q_{1}:\R\times \mathcal{B}(\R) \ni (x_1,A) \mapsto \Q_1(x_1;A) \in [0,1]$, $\Q_{1,2}:\R^2\times \mathcal{B}(\R)\ni ((x_1,x_2),A) \mapsto \Q_{1,2}(x_1,x_2;A) \in [0,1]$ such that 
\begin{align*}
\mu_2 &= \int_\R \int_{\cdot} \Q_1(x_1;\D x_2) \mu_1(\D x_1),\\
\mu_3 &= \int_\R \int_\R \int_{\cdot} \Q_{1,2}(x_1,x_2;\D x_3)\Q_1(x_1;\D x_2) \mu_1(\D x_1), \\
x_2 &= \int_{\R} x_3 ~\Q_{1,2}(x_1,x_2;\D x_3)  ~ \text{ for all }x_1,x_2 \in \R,\\
x_1 &= \int_{\R} x_2 ~\Q_1(x_1;\D x_2) ~\text{ for all } x_1 \in \R, \\
\end{align*}
and such that 
\[\int_{\cdot} \int_{\R}\int_{\cdot} \Q_{1,2}(x_1,x_2;\D x_3)\Q_1(x_1;\D x_2) \mu_1(\D x_1) \in \mathcal{Q}_c^*(\mu_1,\mu_3).
\]
\end{itemize}
\end{prop}


Note that if we consider the negation of the assertion of Proposition~\ref{prop_intermediate_projection}~(i), then improved price bounds can be explained from two different points of view.
While sub-points (ii)--(v) from Proposition~\ref{prop_intermediate_projection} allow to understand the improvement to be induced by the additional flexibility with respect to the semi-static hedging strategies, as discussed at the end of Section~\ref{sec_setting}, sub-points (vi)--(viii) explain the improvement by the exclusion of martingale models that were optimal in the two-marginal case and which are no more consistent with the new set of martingale measures.
Proposition~\ref{prop_intermediate_projection} makes it apparent that indeed both effects occur simultaneously.

{ With the following result we further characterize the set of intermediate marginals preventing improved price bounds.
\begin{prop}\label{prop_characterization_mu_2}
Let $c \in \mathcal{C}_{\operatorname{lin}}(\R^2)$ and assume that Assumption~\ref{asu_marginals} is fulfilled for the marginal distributions $\mu_1,\mu_3 \in \mathcal{P}(\R)$. Consider the set 
\[
\mathcal{I}:= \left\{\mu_2 \in \mathcal{P}(\R)~\middle|~ \mu_1 \preceq \mu_2 \preceq \mu_3\text{ and } \inf_{\Q \in \mathcal{M}(\mu_1,\mu_2,\mu_3)}\E_\Q[c(S_{t_1},S_{t_3})]=\inf_{\Q \in \mathcal{M}(\mu_1,\mu_3)}\E_\Q[c(S_{t_1},S_{t_3})]\right\}.
\]
Then, the following holds true.
\begin{itemize}
\item[(i)] The set $\mathcal{I}$ is closed in the space 
$
\left\{\mu_2 \in \mathcal{P}(\R)~\middle|~ \mu_1 \preceq \mu_2 \preceq \mu_3\right\}
$
equipped with the topology induced by the Wasserstein-distance\footnote{{ For any $\mu,\nu \in \mathcal{P}(\R^d)$ the Wasserstein distance of order $1$ (or Wasserstein $1$-distance) is defined as 
$
W_1(\mu,\nu):=\inf_{\PP \in \Pi(\mu,\nu)}\int_{\R^d \times \R^d} \|x-y\| ~\PP(\D x, \D y),
$
where $\|\cdot\|$ denotes the Euclidean norm on $\R^d$, and where $\Pi(\mu,\nu)$ denotes the set of joint distributions of $\mu$ and $\nu$, compare also for example \cite[Definition 6.1.]{villani2009optimal}. }} of order $1$.
\item[(ii)] The set $\mathcal{I}$ is convex.
\end{itemize}
\end{prop}
\begin{rem}
 As a consequence of Proposition~\ref{prop_characterization_mu_2}~(i), the set 
\[
\left\{\mu_2 \in \mathcal{P}(\R)~\middle|~ \mu_1 \preceq \mu_2 \preceq \mu_3\text{ and } \inf_{\Q \in \mathcal{M}(\mu_1,\mu_2,\mu_3)}\E_\Q[c(S_{t_1},S_{t_3})]>\inf_{\Q \in \mathcal{M}(\mu_1,\mu_3)}\E_\Q[c(S_{t_1},S_{t_3})]\right\}
\]
is open in the space 
$
\left\{\mu_2 \in \mathcal{P}(\R)~\middle|~ \mu_1 \preceq \mu_2 \preceq \mu_3\right\}
$
equipped with the topology induced by the Wasserstein-distance of order $1$, i.e., 
 if an intermediate marginal leads to improved price bounds, then pertubating the marginal slightly w.r.t.\,the Wasserstein-distance while maintaining the convex order still leads to improved price bounds
 \end{rem}
}

\subsection{Degree of Improvement}\label{sec_degree_improvement}

It turns out that the function $H_{(u_i),(\Delta_i)}$, defined in \eqref{eq_defn_H}, is crucial to describe the degree of improvement emerging through the consideration of an additional intermediate marginal. The function $H_{(u_i),(\Delta_i)}$ can be understood as the \emph{sub-hedging error} of the semi-static strategy $(u_i)_{i=1,3},(\Delta_i)_{i=1,2}$, where by Lemma~\ref{lem_positivity_of_H} a pointwise positive value of $H_{(u_i),(\Delta_i)}$ indicates that the strategy is indeed a model-free sub-hedging strategy of $c$.

First we observe that by using $H_{(u_i),(\Delta_i)}$ we obtain an alternative representation of the optimal sub-replication strategies as follows.

\begin{lem}\label{lem_equality_f_dualities_H}
Let $c \in \Clin(\R^2)$ and assume that Assumption~\ref{asu_marginals} is fulfilled for $\mu_1,\mu_2,\mu_3 \in \mathcal{P}(\R)$. Then, we have
\begin{equation}
\label{eq_lem_equality_f_dualities_H_0}
\begin{aligned}
&\sup_{u_i \in L^1(\mu_i),\Delta_i \in \mathcal{C}_b(\R^i)}\bigg\{\sum_{i=1}^3\E_{\mu_i}[u_i]~\bigg|~\sum_{i=1}^3 u_i(x_i)+\Delta_1(x_1)(x_{2}-x_1)+\Delta_2(x_1,x_3)(x_{3}-x_2) \\
&\hspace{8cm}\leq c(x_1,x_3)~\text{ for all } (x_1,x_2,x_3) \in \R^3\bigg\}  \\
=&\sup_{u_1, u_3\in  \mathcal{C}_b(\R)\atop \Delta_1 \in \mathcal{C}_b(\R),\Delta_2 \in \mathcal{C}_b(\R^2)}\bigg\{\E_{\mu_1}[u_1]+\E_{\mu_2}\left[H_{(u_i),(\Delta_i)}\right]+\E_{\mu_3}[u_3]~\bigg\}.
\end{aligned}
\end{equation}
\end{lem}
The following proposition captures the degree of improvement of $\inf_{\Q \in \mathcal{M}(\mu_1,\mu_2,\mu_3)}\E_\Q[c(S_{t_1},S_{t_3})]$ over $\inf_{\Q \in \mathcal{M}(\mu_1,\mu_3)}\E_\Q[c(S_{t_1},S_{t_3})]$.

\begin{prop}\label{prop_degree_of_improvement}
Let $c \in \Clin(\R^2)$ and assume that Assumption~\ref{asu_marginals} is fulfilled for $\mu_1,\mu_2,\mu_3 \in \mathcal{P}(\R)$. Then, we have
\begin{equation}\label{eq_prop_degree_of_improvement_statement}
\begin{aligned}
\inf_{\Q \in \mathcal{M}(\mu_1,\mu_2,\mu_3)}&\E_\Q[c(S_{t_1},S_{t_3})]-\inf_{\Q \in \mathcal{M}(\mu_1,\mu_3)}\E_\Q[c(S_{t_1},S_{t_3})]\\
&\hspace{-2cm}=\inf_{u_1,u_3 \in \C_b(\R) \atop \Delta_1 \in \mathcal{C}_b(\R),\Delta_2 \in \mathcal{C}_b(\R^2): H_{(u_i),(\Delta_i)} \geq 0}\sup_{v_1,v_3 \in \C_b(\R)\atop \widetilde{\Delta}_1 \in \mathcal{C}_b(\R),\widetilde{\Delta}_2 \in \mathcal{C}_b(\R^2)}\left(\E_{\mu_2} \left[H_{(u_i-v_i),(\Delta_i-\widetilde{\Delta}_i)}\right]- \E_{\mu_1}[v_1]- \E_{\mu_3}[v_3]  \right).
\end{aligned}
\end{equation}
\end{prop}

Note that the latter Proposition~\ref{prop_degree_of_improvement} asserts that in the case of improvement, i.e., if we have $\inf_{\Q \in \mathcal{M}(\mu_1,\mu_2,\mu_3)}\E_\Q[c(S_{t_1},S_{t_3})]>\inf_{\Q \in \mathcal{M}(\mu_1,\mu_3)}\E_\Q[c(S_{t_1},S_{t_3})]$, for all sub-hedging strategies with static option trading at two times, i.e., for all  $u_1,u_3 \in \C_b(\R), \Delta_1 \in \mathcal{C}_b(\R),\Delta_2 \in \mathcal{C}_b(\R^2)$ with $H_{(u_i),(\Delta_i)} \geq 0$, we can improve the value of this strategy by considering a strategy  $u_1-v_1,u_3-v_3 \in \C_b(\R), \Delta_1-\widetilde{\Delta}_1 \in \mathcal{C}_b(\R),\Delta_2 -\widetilde{\Delta}_2 \in \mathcal{C}_b(\R^2)$ and investing the remainder of the \emph{gap between $c$ and this strategy} in a static option with payoff $H_{(u_i-v_i),(\Delta_i-\widetilde{\Delta}_i)}(S_{t_2})$ at maturity $t_2$. Indeed, according to Lemma~\ref{lem_positivity_of_H}~(ii) this modified strategy fulfils for all $v_1,v_3 \in \C_b(\R)$ that
\begin{equation}\label{eq_improve_ui_vi}
\begin{aligned}
u_1(x_1)&-v_1(x_1)+H_{(u_i-v_i),(\Delta_i-\widetilde{\Delta}_i)}(x_2)+u_3(x_3)-v_3(x_3)\\
+&\left(\Delta_1(x_1)-\widetilde{\Delta}_1(x_1)\right)(x_2-x_1)+ \left(\Delta_2(x_1,x_2)-\widetilde{\Delta}_2(x_1,x_2)\right)(x_3-x_2) \leq c(x_1,x_3)
\end{aligned}
\end{equation}
for all $(x_1,x_2,x_3) \in \R^3$. Conversely, if we can improve any arbitrary strategy in this sense, then according to Proposition~\ref{prop_degree_of_improvement} we obtain improved price bounds.


\begin{rem}[Improvement of the upper bound]\label{rem_upper_bound}
Note that, by using the relation $\inf_x f(x)= - \sup_{x}-f(x)$, we obtain analogue results as in Propositions~\ref{prop_intermediate_projection} and~\ref{prop_degree_of_improvement} when considering the upper price bound problem. In particular, it holds under the assumptions of Proposition~\ref{prop_degree_of_improvement} that
\begin{equation}\label{eq_prop_degree_of_improvement_statement_upper}
\begin{aligned}
\sup_{\Q \in \mathcal{M}(\mu_1,\mu_3)}&\E_\Q[c(S_{t_1},S_{t_3})]-\sup_{\Q \in \mathcal{M}(\mu_1,\mu_2,\mu_3)}\E_\Q[c(S_{t_1},S_{t_3})]\\
&\hspace{-2cm}=\inf_{\Q \in \mathcal{M}(\mu_1,\mu_2,\mu_3)}\E_\Q[-c(S_{t_1},S_{t_3})]-\inf_{\Q \in \mathcal{M}(\mu_1,\mu_3)}\E_\Q[-c(S_{t_1},S_{t_3})]\\
&\hspace{-2cm}=\inf_{u_1,u_3 \in \C_b(\R) \atop \Delta_1 \in \mathcal{C}_b(\R),\Delta_2 \in \mathcal{C}_b(\R^2): \overline{H}_{(u_i),(\Delta_i)} \geq 0}\sup_{v_1,v_3 \in \C_b(\R)\atop \widetilde{\Delta}_1 \in \mathcal{C}_b(\R),\widetilde{\Delta}_2 \in \mathcal{C}_b(\R^2)}\left(\E_{\mu_2} \left[\overline{H}_{(u_i-v_i),(\Delta_i-\widetilde{\Delta}_i)}\right]- \E_{\mu_1}[v_1]- \E_{\mu_3}[v_3]  \right),
\end{aligned}
\end{equation}
where for $u_1 \in \mathcal{C}_b(\R)$, $u_3 \in \mathcal{C}_b(\R)$, $\Delta_1 \in \mathcal{C}_b(\R)$ and $\Delta_2 \in \mathcal{C}_b(\R^2)$ the function $\overline{H}_{(u_i),(\Delta_i)}$ is defined as
\begin{equation}\label{eq_defn_H_overline}
\begin{aligned}
\R \ni x_2 \mapsto \overline{H}_{(u_i),(\Delta_i)}(x_2):=\inf_{x_1,x_3 \in \R}\bigg\{ &u_1(x_1)+u_3(x_3)+\Delta_1(x_1)(x_2-x_1)\\
&+\Delta_2(x_1,x_2)(x_3-x_2)-c(x_1,x_3)\bigg\},
\end{aligned}
\end{equation}
i.e., $\overline{H}_{(u_i),(\Delta_i)}$ describes the \emph{pointwise gap} between a model-free super-hedging strategy and the payoff $c$.
\end{rem}

\begin{rem}[Improvement in the case of $n$ marginal distributions]\label{rem_nmarginals}
One may also extend Proposition~\ref{prop_degree_of_improvement} to the case with $n\in \N$ marginal distributions $\mu_1,\dots,\mu_n$ satisfying Assumption~\ref{asu_marginals}. To this end, let $I:=\{i_1,\dots,i_m\} \subseteq \{1,\dots,n\}$ with $m < n$ and $c \in \Clin(\R^m)$. Then, we have
\begin{equation*}
\begin{aligned}
\inf_{\Q \in \mathcal{M}(\mu_1,\dots,\mu_n)}&\E_\Q\left[c\left( \left(S_{t_{i}}\right)_{i \in \mathcal{I}}\right)\right]-\inf_{\Q \in \mathcal{M}((\mu_i)_{i\in \mathcal{I}})}\E_\Q\left[c\left( \left(S_{t_{i}}\right)_{i \in \mathcal{I}}\right)\right]\\
&\hspace{-2cm}=\inf_{u_i\in \C_b(\R), i = 1,\dots,m \atop \Delta_i \in \mathcal{C}_b(\R^i), i =1,\dots,n}  \sup_{v_i\in \C_b(\R), i = 1,\dots,m \atop \widetilde{\Delta}_i \in \mathcal{C}_b(\R^i), i =1,\dots,n}  \bigg\{\frac{1}{n-m} \sum_{i \in \{1,\dots,n\} \backslash \mathcal{I}}\E_{\mu_i}\left[H_{(u_j-v_j),(\Delta_j-\widetilde{\Delta}_j)}^{(i)}\right] -\sum_{i \in \mathcal{I}} \E_{\mu_i}[v_i]~\text{ s.t.}~\\
&\hspace{6cm} H_{(u_j-v_j),(\Delta_j-\widetilde{\Delta}_j)}^{(i)} \geq 0 \text{ for all } i \in \{1,\dots,n\} \backslash \mathcal{I}\bigg\},
\end{aligned}
\end{equation*}
with $H_{(u_j),(\Delta_j)}^{(i)}:\R \rightarrow\R$ for $i \in \{1,\dots,n\} \backslash \mathcal{I}$, defined by
\[
H_{(u_j),(\Delta_j)}^{(i)}(x_i):= \inf_{(x_j)_{j\in \{1,\dots,n\} \backslash \{i\}}} \bigg\{ c\left((x_{i_j})_{j=1,\dots,m}\right)-\sum_{j=1}^m u_j(x_{i_j}) -\sum_{j=1}^n \Delta_i(x_1,\dots,x_j)(x_{j+1}-x_j)\bigg\}.
\]
describing the \emph{pointwise gap} between the payoff $c$ (depending on $m$ times) and a sub-hedging strategy with static option trading at $m$ times and dynamic trading in the underlying security at $n$ times.
\label{rem_n_marginals}
\end{rem}
%

\subsection{Financial Markets with a Finite Number of Traded Options}\label{sec_financial_application}

In addition to the results from the idealized setting presented in Section~\ref{sec_setting}, where knowledge about entire marginal distributions is assumed, we next study the influence of additional intermediate price information in a more realistic market environment. In a real financial market only a finite amount of option prices can be observed and therefore, in contrast to the previously discussed setting from Section~\ref{sec_mot}, we do not assume to know the structure of the whole marginal distributions (which only can be obtained when observing prices for a \emph{continuum} of strikes via the results from \cite{breeden1978prices}), but the prices of a finite amount of call options maturing at prospective times $t_1,t_2,t_3$ with $t_1<t_2<t_3$, compare also the settings of the super-replication approaches involving a finite amount of options discussed in \cite{acciaio2016model}, \cite{bartl2020pathwise},  \cite{burzoni2017model}, \cite{hou2018robust}, and \cite{neufeld2021model}.

Let  $S_0\in \R$ denote the current spot value of the underlying security. At initial time, we observe for each maturity $t_i \in \{t_1,t_2,t_3\}$ the prices of $m_i+1$ call options with $m_i \in \N$ and with payoff function $\R \ni x \mapsto (x-K_{i,j})_+$ for $j=0,\dots,m_i$, $i=1,2,3$, where $(K_{i,j})_{i=1,\dots,m_i,\atop j=1,2,3} \subset \R$ denotes the set of associated strikes. Under absence of a bid-ask spread and transaction costs\footnote{For an extension of the presented setting that enables to formulate model-free super-hedging in a market with frictions, we refer the reader to \cite[Section 2]{burzoni2016arbitrage}, \cite[Section 3]{cheridito2017duality}, \cite{dolinsky2014robust}, \cite[Section 6.3]{eckstein2021martingale} or \cite[Appendix A.1]{neufeld2021model}.}, we assume that the call option with payoff $x \mapsto (x-K_{i,j})_+$ can at initial time be bought at price $\Pi_{i,j}$ for $j=0,\dots,m_i$, $i=1,2,3$.
Given $d,\lambda_{i,j}, \Delta_0 \in \R, \Delta_i \in \mathcal{C}_b(\R^i)$, we denote the profit of a semi-static trading strategy with trading at times $t_1$ and $t_3$ by
\[
\Psi^{1,3}_{d,(\lambda_{i,j}),(\Delta_i)}(x_1,x_3):=d+\sum_{i=1,3}\sum_{j=0}^{m_i} \lambda_{i,j}\left(x_i-K_{i,j}\right)^++\Delta_0(x_1-S_0)+\Delta_1(x_1)(x_3-x_1),~~ (x_1,x_3) \in \R^2
\]
and by
\begin{align*}
\Psi^{1,2,3}_{d,(\lambda_{i,j}),(\Delta_i)}(x_1,x_2,x_3):=d&+\sum_{i=1}^3\sum_{j=0}^{m_i} \lambda_{i,j}\left(x_i-K_{i,j}\right)^+\\
&+\Delta_0(x_1-S_0)+\Delta_1(x_1)(x_2-x_1)+\Delta_2(x_1,x_2)(x_3-x_2),
\end{align*}
for $(x_1,x_2,x_3) \in \R^3$ the profit of a semi-static trading strategy with  trading at times $t_1,t_2,t_3$.
We fix a payoff function $\R^2 \ni (x_1,x_3) \mapsto c(x_1,x_3)$ and we study the influence of price information on vanilla options with maturity $t_2$ on the model-free data-driven sub-hedging problem
\begin{equation}\label{eq_prob_vanilla_13}
\begin{aligned}
\sup_{d,\lambda_{i,j},\Delta_0\in \R, \Delta_1 \in \mathcal{C}_b(\R)} \left\{ d+\sum_{j=0}^{m_1} \lambda_{1,j}\Pi_{1,j}+\sum_{j=0}^{m_3} \lambda_{3,j}\Pi_{3,j}~\middle|~\Psi^{1,3}_{d,(\lambda_{i,j}),(\Delta_i)} \leq c \right\}.
\end{aligned}
\end{equation}

In particular, we are interested in conditions which guarantee a strictly positive difference between the formulation from \eqref{eq_prob_vanilla_13} and the value of super-hedging that also involves static trading in call options at intermediate time $t_2$.
\begin{equation}\label{eq_prob_vanilla_123}
\begin{aligned}
\sup_{d,\lambda_{i,j},\Delta_0\in \R, \Delta_i \in \mathcal{C}_b(\R^i)}\left\{d+\sum_{i=1}^3\sum_{j=0}^{m_i} \lambda_{i,j}\Pi_{i,j}~\middle|~\Psi^{1,2,3}_{d,(\lambda_{i,j}),(\Delta_i)} \leq c \right\}.
 \end{aligned}
\end{equation}

To solve  the problems formulated in \eqref{eq_prob_vanilla_13} and \eqref{eq_prob_vanilla_123} with MOT-methods, and then to study the degree of improvement of the sub-hedging problem from \eqref{eq_prob_vanilla_123} over  the sub-hedging problem from \eqref{eq_prob_vanilla_13} by Proposition~\ref{prop_intermediate_projection} and \ref{prop_degree_of_improvement}, we impose the following assumption on the set of strikes and observed market prices.

\begin{asu}\label{asu_strikes}
We assume for all $i=1,2,3$ that
\begin{itemize}
\item[(i)] $0=K_{i,0}<K_{i,1}<\cdots<K_{i,m_i}$,
\item[(ii)] $S_0=\Pi_{i,0}\geq \Pi_{i,1}\geq\cdots\geq\Pi_{i,m_i}=0$,
\item[(iii)] The call option prices are convex with respect to the strike, i.e.,
\[
\frac{\Pi_{i,j+1}-\Pi_{i,j}}{K_{i,j+1}-K_{i,j}}-\frac{\Pi_{i,j}-\Pi_{i,j-1}}{K_{i,j}-K_{i,j-1}}\geq 0 \text{ for all } j=0,\dots,m_i-1.
\]
\end{itemize}
\end{asu}
Note that Assumption~\ref{asu_strikes}~(i), (ii), and (iii) can be considered as natural since they correspond to empirical observations and since they exclude arbitrage, see, e.g. \cite[Section 10.1]{hull2003options}. For Assumption (iii) which prevents the existence of static arbitrage opportunities we refer to  \cite{cohen2020detecting} and \cite{cousot2007conditions}.
We now follow the construction from  \cite[Section 3]{bauerle2019consistent}, and define for $i=1,2,3$ marginal distributions $\mu_i^* \in \mathcal{P}(\R)$ as follows\footnote{We denote by $\delta_x$ the  Dirac measure centered on $x \in \R$, i.e., for any measurable set $A \subset \R$ we have $\delta_x(A) = 1$ if $x\in A$ and $0$ else.}
\begin{align*}
\mu_i^*:= &\sum_{j=0}^{m_i} \left(\frac{\Pi_{i,j+1}-\Pi_{i,j}}{K_{i,j+1}-K_{i,j}}-\frac{\Pi_{i,j}-\Pi_{i,j-1}}{K_{i,j}-K_{i,j-1}}\right)\delta_{K_{i,j}},\\
\text{ with } &\frac{\Pi_{i,m_i+1}-\Pi_{i,m_i}}{K_{i,m_i+1}-K_{i,m_i}}:=0,~
\text{ and }\frac{\Pi_{i,0}-\Pi_{i,-1}}{K_{i,0}-K_{i,-1}}:=-1.
\end{align*}
In particular, the marginals $\mu_i^*$, $i=1,2,3$ are due to Assumption~\ref{asu_strikes}~(iii) well defined probability measures and have the property to be consistent with the observed market prices, i.e., it holds
\[
\E_{\mu_i^*}[(S_{t_i}-K_{i,j})^+]=\Pi_{i,j}\text{ for all } j=0,\dots,m_i, ~~i =1,2,3,
\]
compare also \cite{bauerle2019consistent} and \cite{davis2007range}. Since by construction the  marginals $(\mu_i^*)_{i=1,2,3}$ possess the same mean $S_0$, Assumption~\ref{asu_marginals} is for the marginals $(\mu_i^*)_{i=1,2,3}$  fulfilled if and only if $\E_{\mu_3^*}[(S_{t_3}-K)^+]\geq \E_{\mu_2^*}[(S_{t_2}-K)^+]\geq \E_{\mu_1^*}[(S_{t_1}-K)^+]$ for all $K\in \R$, compare, e.g., \cite{shaked2007stochastic}. Hence, if $(K_1,j)_{j=0,\dots,m_1}=(K_{2,j})_{j=0,\dots,m_2}=(K_{3,j})_{j=0,\dots,m_3}$, then Assumption~\ref{asu_marginals} is fulfilled if $\Pi_{1,j}\leq \Pi_{2,j} \leq \Pi_{3,j}$ for all $j=0,\dots,m_1$. To link the problems  \eqref{eq_prob_vanilla_123} and \eqref{eq_prob_vanilla_13} to corresponding MOT problems we reduce the super-hedging problems  to the grid spanned by the considered strikes $(K_{i,j})_{j=0,\dots,m_i \atop i =1,2,3}$. This leads to the following problems.
\begin{equation}\label{eq_prob_vanilla_13_refined}
\begin{aligned}
{\underline{P}}_{1,3}(c):&=\sup_{d,\lambda_{i,j},\Delta_0\in \R, \atop  \Delta_1 \in \mathcal{C}_b(\R)} \bigg\{ d+\sum_{i=1,3}\sum_{j=0}^{m_i} \lambda_{i,j}\Pi_{i,j}~\bigg|~\Psi^{1,3}_{d,(\lambda_{i,j}),(\Delta_i)}(x_1,x_2,x_3) \leq c(x_1,x_3)  \\ &\hspace{6.5cm}\text{ for all } x_i \in  (K_{i,j})_{j=0,\dots,m_i}, \text{ for } i=1,2,3 \bigg\},\\
{\underline{P}}_{1,2,3}(c):&=\sup_{d,\lambda_{i,j},\Delta_0\in \R, \atop  \Delta_i \in \mathcal{C}_b(\R^i)}\bigg\{d+\sum_{i=1,2,3}\sum_{j=0}^{m_i} \lambda_{i,j}\Pi_{i,j}~\bigg|~\Psi^{1,2,3}_{d,(\lambda_{i,j}),(\Delta_i)}(x_1,x_2,x_3) \leq c(x_1,x_3) \\ &\hspace{6.5cm}\text{ for all } x_i \in  (K_{i,j})_{j=0,\dots,m_i}, \text{ for } i=1,2,3 \bigg\}.
 \end{aligned}
\end{equation}
The above formulated sub-replication problems allow to establish the following result.

\begin{prop}\label{prop_mot_finite_options}
Let Assumption~\ref{asu_marginals} hold true for the marginals $\mu_1^*,\mu_2^*,\mu_3^*$, let Assumption \ref{asu_strikes} be fulfilled and let $c \in \Clin(\R^2)$. Then, we have
\[
\underline{P}_{1,3}(c) = \inf_{\Q\in \mathcal{M}(\mu_1^*,\mu_3^*)}\E_\Q[c(S_{t_1},S_{t_3})]
\]
as well as 
\[
\underline{P}_{1,2,3}(c) = \inf_{\Q\in \mathcal{M}(\mu_1^*,\mu_2^*,\mu_3^*)}\E_\Q[c(S_{t_1},S_{t_3})].
\]
\end{prop}


Note that the MOT-problems from Proposition~\ref{prop_mot_finite_options} can be solved efficiently with finite linear programming methods since they involve discrete marginals, compare also \cite{eckstein2021robust}, \cite{guo2019computational} and \cite{henry2013automated} for a description and analysis of these linear programming methods.

Moreover, by Proposition~\ref{prop_mot_finite_options}, we are able to apply Proposition~\ref{prop_intermediate_projection} and Proposition~\ref{prop_degree_of_improvement} to study the improvement of $\underline{P}_{1,2,3}(c)$ over $\underline{P}_{1,3}(c)$. In particular, we can describe the degree of improvement by the following corollary.

\begin{cor}\label{cor_improvement_call}
Let Assumption~\ref{asu_marginals} hold true for the marginals $\mu_1^*,\mu_2^*,\mu_3^*$, let Assumption \ref{asu_strikes} be fulfilled and let $c \in \Clin(\R^2)$. Then, we have that
\begin{align*}
&\underline{P}_{1,2,3}(c)-\underline{P}_{1,3}(c)\\&=\inf_{d,\lambda_{i,j}, \Delta_0\in \R, \Delta_i \in \mathcal{C}_b(\R^i): \atop  
\Psi^{1,2,3}_{d,(\lambda_{i,j}),(\Delta_i)} \leq c, \lambda_{2,j} \equiv 0} \sup_{\widetilde{d}, \widetilde{\lambda}_{i,j}, \Delta_0\in \R, \widetilde{\Delta}_i \in \mathcal{C}_b(\R^i): \atop \Psi^{1,2,3}_{{d-\widetilde{d}},({\lambda}_{i,j}- \widetilde{\lambda}_{i,j}),(\Delta_i-{\widetilde{\Delta}}_i)} \leq c}
\bigg( d-\widetilde{d}-\sum_{j=0}^{m_2}\widetilde{\lambda}_{2,j}\Pi_{2,j}-\sum_{i=1,3}\sum_{j=0}^{m_i}\left(\lambda_{i,j}-\widetilde{\lambda}_{i,j}\right)\Pi_{i,j} \bigg)\\
&=\inf_{u_1,u_3 \in \C_b(\R) \atop \Delta_1 \in \mathcal{C}_b(\R),\Delta_2 \in \mathcal{C}_b(\R^2): H_{(u_i),(\Delta_i)} \geq 0}\sup_{v_1,v_3 \in \C_b(\R)\atop \widetilde{\Delta}_1 \in \mathcal{C}_b(\R),\widetilde{\Delta}_2 \in \mathcal{C}_b(\R^2)}\left(\E_{\mu_2^*} \left[H_{(u_i-v_i),(\Delta_i-\widetilde{\Delta}_i)}\right]- \E_{\mu_1^*}[v_1]- \E_{\mu_3^*}[v_3]  \right).
\end{align*}
\end{cor}

\section{Examples and Numerical Experiments}\label{sec_examples}

Given two marginal distributions $\mu_1, \mu_3 \in \mathcal{P}(\R)$ with $\mu_1 \preceq \mu_3$ as well as a payoff function $c\in \mathcal{C}_{\operatorname{lin}}(\R^2)$, we know from Proposition~\ref{prop_intermediate_projection} that the relation 
\begin{equation}\label{eq_inf_eq_inf}
\inf_{\Q \in \mathcal{M}(\mu_1,\mu_2,\mu_3)}\E_\Q[c(S_{t_1},S_{t_3})]=\inf_{\Q \in \mathcal{M}(\mu_1,\mu_3)}\E_\Q[c(S_{t_1},S_{t_3})]
\end{equation} is fulfilled if and only if one of the requirements of Proposition~\ref{prop_intermediate_projection} is met. In Section~\ref{sec_no_improvement}, we discuss two classes of marginal distributions and associated payoff functions which guarantee equation  \eqref{eq_inf_eq_inf}. Moreover, in Section~\ref{sec_improvement} we study examples where tighter price bounds can be observed. Eventually, with Example~\ref{exa_real} we provide an investigation of the problem when applied to real financial markets, relying on the findings from Section~\ref{sec_financial_application}.
\subsection{Examples for the Exclusion of Improvement} \label{sec_no_improvement}
\subsubsection{Convex interpolation of Intermediate Marginals}\label{sec_convex_interpol}

Suppose for a fixed function $c\in \mathcal{C}_{\operatorname{lin}}(\R^2)$ and for marginals $\mu_1, \mu_3 \in \mathcal{P}(\R)$ with $\mu_1 \preceq \mu_3$ there exists some measure $\Q^*\in \mathcal{Q}_c^*(\mu_1,\mu_3)$ which exhibits a specific structure as it is supported on two deterministic maps. More specifically, we assume
\begin{equation}\label{eq_definition_q*}
\Q^*(\D x_1,\D x_3):=\mu_1(\D x_1)\left(q(x_1)\delta_{T_u(x_1)}(x_3)+(1-q(x_1))\delta_{T_d(x_1)}(x_3)\right)\D x_3
\end{equation}
with $T_d(x) \leq x \leq T_u(x)$ for all $x \in \R$ for some functions $T_d,T_u:\R\rightarrow \R$, and where
\begin{equation}\label{eq_definition_prob_q}
\R \ni x \mapsto q(x):=\frac{x-T_d(x)}{T_u(x)-T_d(x)}\one_{\{T_u(x)>T_d(x)\}} \in [0,1].
\end{equation}
Solutions of the type as in \eqref{eq_definition_q*} are optimal for a broad class of payoff functions and are for example discussed in  \cite{beiglbock2016problem}, \cite{beiglbock2021shadow} and \cite{henry2016explicit} in great detail.
In \cite{henry2017modelfree}, the author defines an interpolation between marginal measures which preserves the convex order. We use this definition to define a marginal $\mu_2 \in \mathcal{P}(\R)$ with
\[
\mu_1 \preceq \mu_2 \preceq \mu_3.
\]
\begin{asu}[Martingale convex interpolation] \label{asu_intermediate_1}
Assume $\mu_2 = \operatorname{Law}(S_{t_2}) \in \mathcal{P}(\R)$ with $S_{t_2}$ defined by  
\[
S_{t_2}:=\begin{cases}
&S_{t_1}\cdot(1-t)+t\cdot T_u(S_{t_1}) ~\text{ with probability } q(S_{t_1}), \\
&S_{t_1}\cdot(1-t)+t\cdot T_d(S_{t_1}) ~\text{ with probability } 1-q(S_{t_1})
\end{cases}
\]
for some $t \in [0,1]$, where $q$ is defined in \eqref{eq_definition_prob_q}.
\end{asu}
The fact that the marginal distributions $\mu_1,\mu_2,\mu_3$, where $\mu_2$ is defined according to Assumption~\ref{asu_intermediate_1}, increase in convex order is ensured by \cite[Lemma 2.2.]{henry2017modelfree}.
With these definitions we are able to state the following result.

\begin{cor} Let $c\in \mathcal{C}_{\operatorname{lin}}(\R^2)$ and let Assumption~\ref{asu_intermediate_1} be true. Then we have no improved price bounds, i.e.,
\eqref{eq_equality of problems} holds true.
\end{cor}
\begin{proof} To show the assertion we apply Proposition~\ref{prop_intermediate_projection}~(viii). To this end,  we define for all $x_1,x_2 \in \R$
\begin{align}
\Q_{1}(x_1;\D x_2)&:=\left(q(x_1)\delta_{(1-t)x_1+tT_u(x_1)}(x_2)+(1-q(x_1))\delta_{(1-t)x_1+tT_d(x_1)}(x_2)\right)\D x_2, \notag \\
\Q_{1,2}(x_1,x_2;\D x_3)&:=\left(\widetilde{q}(x_1,x_2)\delta_{T_u(x_1)}(x_3)+(1-\widetilde{q}(x_1,x_2))\delta_{T_d(x_1)}(x_3)\right)\D x_3, \notag \\
\text{ with } \widetilde{q}(x_1,x_2)&:=\frac{x_2-T_d(x_1)}{T_u(x_1)-T_d(x_1)}\one_{\{T_u(x_1)>T_d(x_1)\}}. \label{eq_q_tilde_def}
\end{align}
This defines { probability} kernels which fulfil the requirements of Proposition~\ref{prop_intermediate_projection}~(viii) {, compare also the proof of Proposition~\ref{prop_intermediate_Tu_Td}.} 
\end{proof}

\begin{rem}\label{rem_convex_interpol}
The assumption for the martingale convex interpolation in Assumption~\ref{asu_intermediate_1} can be slightly generalized and still one does not obtain improved price bounds, i.e., \eqref{eq_equality of problems} holds true. More precisely, assume 
$\mu_2=\operatorname{Law}(S_{t_2})\in \mathcal{P}(\R)$ with
\[
S_{t_2}=\begin{cases}
f_1(S_{t_1})~\text{ with probability } q(S_{t_1}), \\
f_2(S_{t_1})~\text{ with probability } 1-q(S_{t_1}),
\end{cases}
\]
for some functions $f_1,f_2:\R \rightarrow \R$ 
fulfilling $T_d \leq f_1 \leq \operatorname{Id} \leq f_2 \leq T_u$,
such that the following relation holds for all $x \in \R$:
\begin{equation}\label{eq_asu414_martingale}
q(x)f_1(x)+(1-q(x))f_2(x)=x,
\end{equation}
where $\R \ni x\mapsto q(x)=\frac{x-T_d(x)}{T_u(x)-T_d(x)}\one_{\{T_u(x)>T_d(x)\}}$ depends on the optimal measure $\Q^*\in \mathcal{Q}_c^*(\mu_1,\mu_3)$ as defined in \eqref{eq_definition_q*}. In this case, one defines Markov kernels fulfilling the requirements from Proposition~\ref{prop_intermediate_projection}~(vii) through
\begin{align*}
\Q_{1}(x_1;\D x_2)&:=\left(q(x_1)\delta_{f_1(x_1)}(x_2)+(1-q(x_1))\delta_{f_2(x_1)}(x_2)\right)\D x_2,\\
\Q_{1,2}(x_1,x_2;\D x_3)&:=\left(\widetilde{q}(x_1,x_2)\delta_{T_u(x_1)}(x_3)+(1-\widetilde{q}(x_1,x_2))\delta_{T_d(x_1)}(x_3)\right)\D x_3
\end{align*}
and with $\widetilde{q}$ defined as in \eqref{eq_q_tilde_def}{, compare also the proof of Proposition~\ref{prop_intermediate_Tu_Td}.} 
\end{rem}

\begin{exa}[Uniform marginals, Spence--Mirrlees cost function, Upper Bound]\label{exa_general_sm}
We consider a payoff function
\[
\R^3 \ni (x_1,x_2,x_3) \mapsto c(x_1,x_2,x_3):=x_1(x_3-x_1)^2,
\]
and continuous uniform marginal distributions $\mu_1 = \mathcal{U}([-1,1]), \mu_2 = \mathcal{U}([-2,2])$.
One can show that with
\begin{equation}\label{eq_def_Tu_Td}
\begin{aligned}
&\R \ni x \mapsto T_u(x)= \left(\frac{3}{2}x+\frac{1}{2}\right)\one_{\{x >-1\}}+x\one_{\{x\leq -1\}},\\
&\R \ni x \mapsto T_d(x)=\left(-\frac{1}{2}x-\frac{3}{2}\right)\one_{\{x >-1\}}+x\one_{\{x\leq -1\}}
\end{aligned}
\end{equation}
and $\R \ni x \mapsto q(x):= \frac{3}{4}\one_{\{x > -1\}}(x)$ the measure $\Q^*$ defined in \eqref{eq_definition_q*} fulfils
$
\E_{\Q^*}[c]= \sup_{\Q\in \mathcal{M}(\mu_1,\mu_3)} \E_{\Q}[c]
$, 
compare for example \cite{henry2016explicit}.
\begin{itemize}
\item[(i)]
In Figure \ref{fig_convex1}, we illustrate the intermediate marginal $\mu_2$ as defined by Assumption \ref{asu_intermediate_1} for different values of $t \in [0,1]$.
The density function of $\mu_2$ is given by
\[
\R \ni x \mapsto f_{\mu_2}(x)=\begin{cases}
&\frac{3}{8+4t}\one_{[-1,1+t]}(x)+\frac{1}{8-12t}\one_{[-1,1-3t]}(x) ~\text{ if } t \leq \frac{2}{3},\\
&\frac{3}{8+4t}\one_{[-1,1+t]}(x)-\frac{1}{8-12t}\one_{[1-3t,-1]}(x) ~\text{ if } t > \frac{2}{3}.
\end{cases}
\]
\begin{center}
\begin{figure}[h!]
	\subfigure[The interpolation from Example~\ref{exa_general_sm}~(i).]{
    \includegraphics[width=0.4\textwidth]{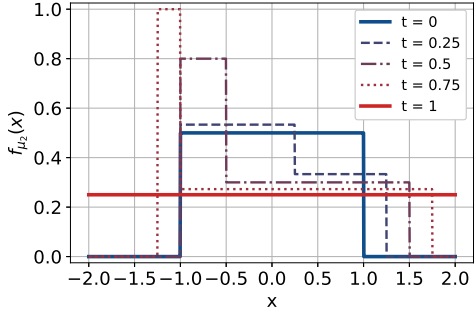} \label{fig_convex1}}
    \hspace{0.1\textwidth}
    \subfigure[The interpolation from Example~\ref{eq_ineq_u1u2u3c}.]{
    \includegraphics[width=0.4\textwidth]{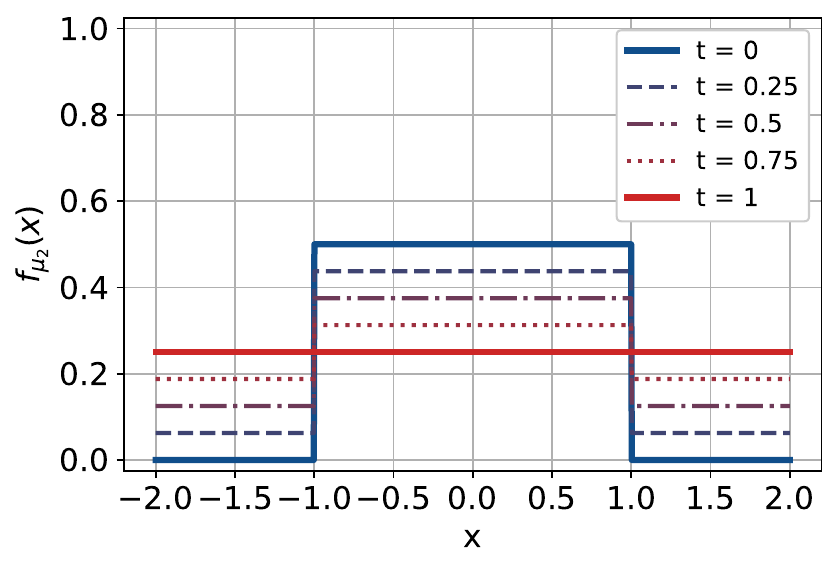}\label{fig_unif1}}
              \caption{The two plots show the probability density functions of interpolated marginals $\mu_2$ with $\mu_1 \preceq \mu_2 \preceq \mu_3$, where $\mu_1= \mathcal{U}([-1,1])$, $\mu_3 = \mathcal{U}([-2,2])$. The density functions of $\mu_2$ from Example~~\ref{exa_general_sm}~(i) (Figure~\ref{fig_convex1}) and Example~\ref{eq_ineq_u1u2u3c} (Figure~\ref{fig_unif1}) are depicted in dependence of different parameters $t\in [0,1]$.} \label{fig_unif_convex}
\end{figure}
\end{center}
\item[(ii)]
Similarly, by following Remark~\ref{rem_convex_interpol}, we can exclude improvement for intermediate measures possessing atoms.
By defining
\[
S_{t_2}= \begin{cases}
\frac{4}{3}S_{t_1}+\frac{1}{3},~&\text{ with probability } \frac{3}{4}, \\
-1,~&\text{ with probability } \frac{1}{4}. \\
\end{cases}
\]
we obtain an intermediate marginal $\mu_2: = \operatorname{Law}(S_{t_2})$
fulfilling the assumption postulated in \eqref{eq_asu414_martingale}.
\end{itemize}

\end{exa}

\subsubsection{Linear Interpolation of Intermediate Marginals}\label{sec_linear_interpol}

Suppose the marginal $\mu_2 \in \mathcal{P}(\R)$ possesses the property that for all $u \in L^1(\mu_2)$ there exists some $\lambda_u \in [0,1]$ such that 
\begin{equation}\label{eq_linear_interpolation}
\E_{\mu_2}[u]= \lambda_u \E_{\mu_1}[u]+(1-\lambda_u)\E_{\mu_3}[u].
\end{equation}
Then, as a consequence of Proposition~\ref{prop_intermediate_projection}, we show that whenever $\mu_2$ fulfils the above equation \eqref{eq_linear_interpolation}, then including this marginal information does not lead to an improved value of the MOT-problem.
\begin{lem}
Let $c\in \Clin(\R^2)$ and let $\mu_2 \in \mathcal{P}(\R)$ be such that \eqref{eq_linear_interpolation} holds true, then we obtain no improvement through inclusion of the marginal $\mu_2$, i.e., \eqref{eq_inf_eq_inf} is valid.
\end{lem}

\begin{proof}
Let $u_1,u_2, u_3, \Delta_1 \in \mathcal{C}_b(\R), \Delta_2 \in \mathcal{C}_b(\R^2)$ be such that 
\begin{equation}\label{eq_ineq_u1u2u3c}
u_1(x_1)+u_2(x_2)+ u_3(x_3)+\Delta_1(x_1)(x_2-x_1)+\Delta_2(x_1,x_2)(x_3-x_2) \leq c(x_1,x_3)\text{ for all } x_1,x_2,x_3 \in \R.
\end{equation}
By assumption \eqref{eq_linear_interpolation} there exists some $\lambda_{u_2} \in [0,1]$ such that 
\begin{equation}\label{eq_u2_lambda_u}
\E_{\mu_2}[u_2]= \lambda_{u_2} \E_{\mu_1}[u_2]+(1-\lambda_{u_2})\E_{\mu_3}[u_2].
\end{equation}
Then, we first set in \eqref{eq_ineq_u1u2u3c} $x_2=x_1$, then $x_2=x_3$ and consider convex combinations of the resulting inequalities to obtain the inequality
\begin{align*}
&u_1(x_1)+\lambda_{u_2} u_2(x_1)+\left(1-\lambda_{u_2}\right) u_2(x_3) +u_3(x_3)\\
&\hspace{2.5cm}+\left(\left(1-\lambda_{u_2}\right)\Delta_1(x_1)+\lambda_{u_2} \Delta_2(x_1,x_1)\right)(x_3-x_1) \leq c(x_1,x_3)\text{ for all } x_1,x_3 \in \R.
\end{align*}
Further, by \eqref{eq_u2_lambda_u} we have that 
\[
\E_{\mu_1}[u_1+\lambda_{u_2} u_2]+\E_{\mu_3}[(1-\lambda_{u_2})u_2+u_3]=\E_{\mu_1}[u_1]+\E_{\mu_2}[u_2]+\E_{\mu_3}[u_3]
\]
and the assertion follows by Proposition~\ref{prop_intermediate_projection}~(v).
\end{proof}

\begin{exa}[Uniform marginals]\label{exa_unif1}
We consider the marginal distributions
\[
\mu_1 =\mathcal{U}([-1,1]),~\mu_3=\mathcal{U}([-2,2]),
\]
and the Spence-Mirrlees type cost function $c(x_1,x_2,x_3)=x_1(x_3-x_1)^2$.
 Define $\mu_2=\operatorname{Law}(S_{t_2})\in \mathcal{P}(\R)$ for $t \in [0,1]$ with 
\[
S_{t_2}:=\begin{cases}
S_{t_1}, &\text{ with probability } 1-t, \\
\frac{3}{2}S_{t_1}+\frac{1}{2}, &\text{ with probability } t\frac{3}{4}, \\
-\frac{1}{2}S_{t_1}-\frac{3}{2}, &\text{ with probability } t\frac{1}{4}.
\end{cases}
\]
Then, we use the notation from Example~\ref{exa_general_sm} and have for all $u \in L^1(\mu_2)$ and for all $t \in [0,1]$ that
\begin{align*}
\E_{\mu_2}[u(S_{t_2})]&=\E_{\mu_1}[(1-t)\cdot u(S_{t_1})]+\E_{\mu_1}\left[t\frac{3}{4}u\left(\frac{3}{2}S_{t_1}+\frac{1}{2}\right)\right]+\E_{\mu_1}\left[t\frac{1}{4}u\left(-\frac{1}{2}S_{t_1}-\frac{3}{2}\right)\right]\\
&=\E_{\mu_1}[(1-t)\cdot u(S_{t_1})]+\E_{\mu_1}\left[tq(S_{t_1})u\left(T_u(S_{t_1})\right)\right]+\E_{\mu_1}\left[t(1-q(S_{t_1}))u\left(T_d(S_{t_1})\right)\right]\\
&=(1-t)\E_{\mu_1}[u(S_{t_1})]+t\E_{\mu_3}[u(S_{t_3})].
\end{align*}
and Assumption \eqref{eq_linear_interpolation} is fulfilled.
Compare Figure \ref{fig_unif1} for an illustration of the resulting marginal density functions when varying the parameter $t \in [0,1]$. The probability density functions of the intermediate marginal distributions are given by
\[
\R \ni x \mapsto f_{\mu_2}(x)=t\frac{1}{4}\one_{[-2,2]\backslash [-1,1]}(x)+\left(t\frac{1}{4}+(1-t)\frac{1}{2}\right)\one_{[-1,1]}(x),~t \in [0,1].
\]
\end{exa}

\subsection{Examples for Improvement}\label{sec_improvement}
{
\subsubsection{The case $T_d \leq \operatorname{Id} \leq T_u$}
In Section~\ref{sec_convex_interpol} we have seen that \emph{interpolating} the optimal martingale coupling at an intermediate time implies marginal distributions preventing improved price bounds. 
By formalizing these observations below, we entirely characterize marginal distributions leading to improved price bounds given the setting of Section~\ref{sec_convex_interpol}.
\begin{prop}\label{prop_intermediate_Tu_Td}
Let $c\in \mathcal{C}_{\operatorname{lin}}(\R^2)$ and $\mu_1,\mu_2, \mu_3 \in \mathcal{P}(\R)$ with $\mu_1 \preceq \mu_2 \preceq \mu_3$. Assume $\mathcal{Q}_c^*(\mu_1,\mu_3)=\{\Q^*\}$ for $\Q^*$ defined in \eqref{eq_definition_q*}. Then, we have
$
\inf_{\Q \in \mathcal{M}(\mu_1,\mu_2,\mu_3)}\E_\Q[c(S_{t_1},S_{t_3})]=\inf_{\Q \in \mathcal{M}(\mu_1,\mu_3)}\E_\Q[c(S_{t_1},S_{t_3})]
$ if and only if there exists some $\widetilde{\Q} \in \mathcal{M}(\mu_1,\mu_2)$ with 
$$
\widetilde{\Q} \big(T_d(S_{t_1}) \leq S_{t_2} \leq T_u(S_{t_1})\big)=1.
$$
\end{prop}
The above Proposition~\ref{prop_intermediate_Tu_Td}  asserts that in terms of Proposition~\ref{prop_characterization_mu_2}, the set of intermediate marginals leading to improved price bounds is given by 
\[
\left\{\mu_2 \in \mathcal{P}(\R)~\middle|~ \mu_1 \preceq \mu_2 \preceq \mu_3 ,~\not\exists \Q \in \mathcal{M}(\mu_1,\mu_2): ~{\Q} \big(T_d(S_{t_1}) \leq S_{t_2} \leq T_u(S_{t_1})\right)=1 \big\}.
\]}
{
\begin{exa}
We consider the marginal distributions
$
\mu_1=\mathcal{U}([-1,1]),~\mu_2=\mathcal{U}(\{-1,1\}),~\mu_3=\mathcal{U}([-2,2])
$
and a payoff function
$
\R^3 \ni (x_1,x_2,x_3) \mapsto c(x_1,x_2,x_3)=x_1(x_3-x_1)^2.
$
Since the partial derivatives of $c$ fulfil the relation $\frac{\partial^3c}{\partial {x_1} \partial x_3^2}=2>0$, the \emph{left-curtain} coupling is according to \cite{henry2016explicit} and \cite{beiglboeck2016problem} the unique optimizing measure for the (upper bound) two-marginal martingale transport problem. More precisely, it holds
\[
\sup_{\Q \in \mathcal{M}(\mu_1,\mu_3)} \E_\Q[c(S_{t_1},S_{t_3})]=\E_{\Q_{\operatorname{lc}}}[c(S_{t_1},S_{t_3})]
\]
for the joint distribution $\Q_{\operatorname{lc}}$ defined via
\[
\Q_{\operatorname{lc}}(\D x_1,\D x_3)=\mu_1(dx_1)\left(q(x_1)\delta_{T_u(x_1)}(x_3)+(1-q(x_1))\delta_{T_d(x_1)}(x_3)\right)(\D x_3),
\]
where, for the prevailing marginal distributions, the functions $T_d, T_u$ are  given by \eqref{eq_def_Tu_Td}.
Let $\Q\in \mathcal{M}(\mu_1,\mu_2)$ be arbitrary, then the martingale property 
\[
S_{t_1} = \Q(S_{t_2}=1~|~S_{t_1})-\Q(S_{t_2}=-1~|~S_{t_1})
\]
implies together with $\Q(S_{t_2}=1~|~S_{t_1})+\Q(S_{t_2}=-1~|~S_{t_1})=1$ that
\[
\Q(S_{t_2}=1~|~S_{t_1}) =\frac{1}{2}(1+S_{t_1}),\qquad\Q(S_{t_2}=-1~|~S_{t_1}) =\frac{1}{2}(1-S_{t_1}).
\]
We then see $\Q(S_{t_2}=1~|~S_{t_1}=0) =\tfrac{1}{2}>0$, and $T_u(0) = \frac{1}{2}$ leading to $\Q\left(S_{t_2}>T_u(S_{t_1})\right) >0$. Hence, according to Proposition~\ref{prop_intermediate_Tu_Td} the price bounds can be improved. Indeed, one can verify that 
\[
0=\sup_{\Q \in \mathcal{M}(\mu_1,\mu_2,\mu_3)}\E_\Q[c(S_{t_1},S_{t_3})]<\sup_{\Q \in \mathcal{M}(\mu_1,\mu_3)}\E_\Q[c(S_{t_1},S_{t_3})]= \E_{\Q_{\operatorname{lc}}}[c(S_{t_1},S_{t_3})] =0.5.
\]
\end{exa}
}

\subsubsection{Further Examples}
In this section we discuss several examples for improved price bounds. While the first two settings discussed in Example~\ref{exa_theoretical_1} and Example~\ref{exa_binomial} rely on the knowledge of entire marginal distributions, we provide with Example~\ref{exa_artificial_calls}  and Example~\ref{exa_real} also an investigation of the situation in practice where only a finite amount of prices of traded options can be taken into account.

\begin{exa}\label{exa_theoretical_1}
We consider the three marginal distributions
\[
\mu_1 = \mathcal{U}([-1,1]),~\mu_2 = \mathcal{U}(\{-1,1,\}),~\mu_3 = \mathcal{U}([-2,2]),
\]
and the payoff function of a forward start straddle, given by
\[
\R^2 \ni (x_1,x_3) \mapsto c(x_1,x_3):=|x_1-x_3|.
\]
{ First note that in this situation Proposition~\ref{prop_intermediate_Tu_Td} is not applicable since, as shown in \cite{hobson2015robust} or \cite{ghoussoub2019structure}, conditional on $S_{t_1}$ the law of $S_{t_3}$ is supported on three and not only on two values, under the optimal martingale measure for the minimization problem\footnote{{As shown in \cite{hobson2012robust}, under the unique optimal measure for the \textbf{maximization} problem, the conditional law of $S_{t_3}$ is  supported only on two values, and therefore Proposition~\ref{prop_intermediate_Tu_Td} is applicable to the analogue maximization problem.}}}.
We show that in this setting including the marginal $\mu_2$ { nevertheless} improves the price bounds. In \cite{hobson2015robust} model-independent price bounds for this specific derivative were extensively studied, and it was shown that the solution of $\inf_{\Q \in \mathcal{M}(\mu_1,\mu_3)}\E_\Q[c]$ is determined via the dual strategy given by
\begin{align*}
 u_3(x)&:= \left[\alpha(p^{-1}(x))+(p^{-1}(x)-x)(1-\theta(p^{-1}(x)))\right]\one_{\{x<-1\}} +\alpha(x) \one_{\{x\in (-1,1)\}}\\
  &\hspace{1cm}\left[+\alpha(q^{-1}(x))
+q^{-1}(x)-x)(-1-\theta(q^{-1}(x)))\right]\one_{\{x>1\}},\\
 u_1(x)&:= - u_3(x),\\
\Delta_1(x)&:= -\theta(p^{-1}(x))\one_{\{x<-1\}}-\theta(x) \one_{\{x \in (-1,1)\}}-\theta)q^{-1}(x)\one_{\{x>1\}},
\end{align*}
with 
\begin{align*}
\alpha(x):&= ((2x)/\sqrt{3})(\arcsin(x/2))+(2-\sqrt{4-x^2})/\sqrt{3},&& \theta(x):= \frac{2}{\sqrt{3}} \arcsin(x/2),\\
p(x):&=\frac{-\sqrt{12-3x^2}-x}{2},  && p^{-1}(x):= \frac{-x-\sqrt{3(4-x^2)}}{2}, \\
q(x):&=\frac{\sqrt{12-3x^2}-x}{2},    && q^{-1}(x):=\frac{-x+\sqrt{3(4-x^2)}}{2},
\end{align*}
for all $x\in \R$. This leads to a value of $\inf_{\Q \in \mathcal{M}(\mu_1,\mu_3)}\E_\Q[c] =\E_{\mu_1}[u_1]+\E_{\mu_3}[u_3] \approx 0.5931$. We consider now a strategy
\begin{align*}
\widehat{u}_1(x_1)&:=\left(1-x_1\right)^2 \one_{\{x_1\in [-1,1]\}},&& \widehat{u}_2(x_2):=\left(-1-|x_2|\right)\one_{\{x_2 \not\in \{-1,1\}\}},&&\widehat{u}_3:\equiv 0, \\
\widehat{\Delta}_1(x_1)&:=-x_1\one_{\{x_1\in [-1,1]\}},&& \widehat{\Delta}_2(x_1,x_2):=\left(\one_{\{x_2=1\}}-\one_{\{x_2=-1\}}\right) \one_{\{x_1\in [-1,1]\}}
\end{align*}
fulfilling for all $ x_1,x_2,x_3 \in \R$ that
\begin{align*}
&\widehat{u}_1(x_1)+\widehat{u}_2(x_2)+\widehat{u}_3(x_3) +\widehat{\Delta}_1(x_1)(x_2-x_1)+\widehat{\Delta}_2(x_1,x_2)(x_3-x_2) \\
=& \one_{\{x_1\in [-1,1]\}}\bigg(\left(-1-|x_2|-x_1 x_2 \right)\one_{\{x_2 \not\in \{-1,1\}\}}  \\
&\hspace{2.5cm}+\left(1-x_1+x_3-1\right) \one_{\{x_2=1\}}\\
&\hspace{2.5cm}+\left(1+x_1-x_3-1\right) \one_{\{x_2=-1\}}
\bigg)
\leq |x_1-x_3|.
\end{align*}
The value of this dual sub-replication strategy computes as
$\E_{\mu_1}[\widehat{u}_1]+\E_{\mu_2}[\widehat{u}_2]+\E_{\mu_3}[\widehat{u}_3]=\E_{\mu_1}[\widehat{u}_1]=\frac{2}{3}$. In particular, we have 
$$
\inf_{\Q \in \mathcal{M}(\mu_1,\mu_2, \mu_3)}\E_\Q[c] \geq \inf_{\Q \in \mathcal{M}(\mu_1,\mu_2, \mu_3)}\E_\Q[\widehat{u}_1+\widehat{u}_2+\widehat{u}_3] = \E_{\mu_1}[\widehat{u}_1]=\frac{2}{3} > \inf_{\Q \in \mathcal{M}(\mu_1,\mu_3)}\E_\Q[c],
$$
demonstrating that we can indeed improve the price bounds by including the intermediate marginal $\mu_2$. 

Alternatively, we can verify the existence of improvement by considering an arbitrary  strategy $u_1,u_3 \in \C_b(\R), \Delta_1 \in \mathcal{C}_b(\R),\Delta_2 \in \mathcal{C}_b(\R^2)$ with $H_{(u_i),(\Delta_i)} \geq 0$, then defining 
$v_1:= u_1-\widehat{u}_1, v_3 :=  u_3-\widehat{u}_3$, $\widetilde{\Delta}_1 := \Delta_1 - \widehat{\Delta}_1$, $\widetilde{\Delta}_2:= \Delta_2 - \widehat{\Delta}_2$ and eventually checking the condition from Proposition~\ref{prop_intermediate_projection}~(ii) implying 
\begin{align*}
\E_{\mu_2} \left[H_{(u_i-v_i),(\Delta_i-\widetilde{\Delta}_i)}\right]- \E_{\mu_1}[v_1]- \E_{\mu_3}[v_3] &\geq \E_{\mu_2}[\widehat{u}_2]+\frac{2}{3}-\E_{\mu_1}[u_1]-\E_{\mu_3}[u_3] \\
&= \frac{2}{3}-\E_{\mu_1}[u_1]-\E_{\mu_3}[u_3] \geq \frac{2}{3}- 0.5931 >0.
\end{align*}

 Note that the improvement does not rely on the discrete nature of the intermediate marginal. Indeed, when considering for $0<\varepsilon<1$ marginal distributions of the form $\mu_2^\varepsilon = \mathcal{U}(A^{\varepsilon})$ with $A^{\varepsilon}:=A_1^\varepsilon \cup A_2^\varepsilon$ for $A_1^\varepsilon:=[-1,-1+\varepsilon]$ and $A_2^\varepsilon:=[1-\varepsilon,1]$, then we obtain with the sub-hedging strategy 
\begin{align*}
\widehat{u}_1^{\varepsilon}(x_1)&:=\left(1-\varepsilon-x_1\right)^2 \one_{\{x_1\in [-1,1]\}},  &&\widehat{u}_2^{\varepsilon}(x_2):=\left(-1+\varepsilon-|x_2|\right)\one_{\{x_2 \not\in A^{\varepsilon}\}},&& \widehat{u}_3^{\varepsilon}\equiv 0, \\
\widehat{\Delta}_1^{\varepsilon}(x_1)&:=-x_1\one_{\{x_1\in [-1,1]\}},&&\widehat{\Delta}_2^{\varepsilon}(x_1,x_2):=\left(\one_{A_2^{\varepsilon}}-\one_{A_1^{\varepsilon}}\right) \one_{\{x_1\in [-1,1]\}},
\end{align*}
that 
$\E_{\mu_1^{\varepsilon}}[\widehat{u}_1^{\varepsilon}]+\E_{\mu_2^{\varepsilon}}[\widehat{u}_2^{\varepsilon}] +\E_{\mu_3^{\varepsilon}}[\widehat{u}_3^{\varepsilon}]= \frac{2}{3}-\varepsilon$ and therefore an improvement over $\inf_{\Q \in \mathcal{M}(\mu_1,\mu_3)}\E_\Q[c]\approx 0.5931$ whenever $\varepsilon$ is small enough.
\end{exa}

\label{sec_improvement}

We modify an example that was already studied in \cite{eckstein2021martingale}. Note that the computations for all of the subsequent Examples~\ref{exa_binomial}, \ref{exa_artificial_calls}, and \ref{exa_real} rely on linear programming approaches to calculate solutions for MOT-problems with discrete marginals, see also \cite{eckstein2021robust}, \cite{guo2019computational} and \cite{henry2013automated}.

\begin{exa}\label{exa_binomial}
The authors from \cite{eckstein2021martingale} consider marginals possessing the same support as the marginals in an additive binomial model with step size $1$, starting value $S_{t_0}=100$ and nine prospective times $t_1,\dots,t_9$. More specifically, the marginals are supported on 
\begin{equation}\label{eq_support_1}
\{100-i,100-i+2,\cdots,100+i\}~\text{ for } i =1,\dots,9.
\end{equation}
The authors further study uniform marginals on the given support. When considering a financial derivative with payoff function $(S_{t_9}-S_{t_8})_+$ the authors observe no improvement by solely adding additional marginal information on $\mu_1,\dots,\mu_7$ associated to $S_{t_1},\cdots,S_{t_7}$. However, a large amount of improvement is observable if an additional assumption on the homogeneity of an underlying process is incorporated.\\
We modify the setting and consider on the support defined through \eqref{eq_support_1} new marginals implied by an additive binomial model starting at $100$ with step size $1$ in which the probability for an upward movement is $0.5$, i.e. $S_{t_1} \sim \mu_1=\mathcal{U}(\{99,101\}),~S_{t_2} \sim \mu_2 = \frac{1}{4}\delta_{98}+\frac{1}{2}\delta_{100}+\frac{1}{4}\delta_{102}$ etc.

\begin{figure}[h!]
\begin{center}
    \includegraphics[width=0.4\textwidth]{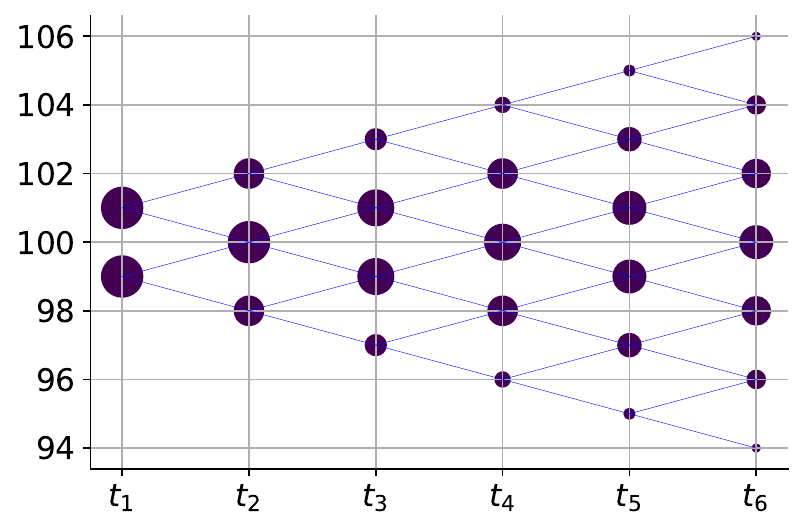} 
    \caption{The marginals in Example~\ref{exa_binomial} are implied by this binomial model, where the size of the dots indicate the probabilities to reach the respective node.}\label{fig_binomial_model}
    \end{center}
\end{figure}
 Compare also Figure~\ref{fig_binomial_model} where we illustrate this model for six time steps and indicate the corresponding probabilities of the supporting values of the marginals through the size of the dots at the respective nodes. Then, the derivative under consideration is $(S_{t_6}-S_{t_1})_+$ and we gradually add information on intermediate marginals in two different ways:
\begin{itemize}
\item[1.]
We start by adding to incorporate information on $\mu_2$, proceed with $\mu_3$, $\mu_4$, and eventually include $\mu_5$. We include information \textbf{from the left} in terms of the time scale.
\item[2.]
We start by adding to incorporate information on $\mu_5$, proceed with $\mu_4$, $\mu_3$, and in the last step we include $\mu_2$. We include information \textbf{from the right}.
\end{itemize}
Both approaches lead stepwise to a remarkable amount of improvement. With all marginals included, lower and upper bound even coincide. However, including information \emph{from the right} tightens the bounds faster. This means, by Proposition~\ref{prop_intermediate_projection}~(vi), that including information on $\mu_5$ restricts possible martingale transport plans more than taking into account information on $\mu_2$. Information on $\mu_2$ alone has barely an impact, whereas in combination with information on $\mu_5$ it leads to some improvement.
\begin{center}
\begin{figure}[h!]
	\subfigure[Including intermediate marginals from the left.]{
    \includegraphics[width=0.45\textwidth]{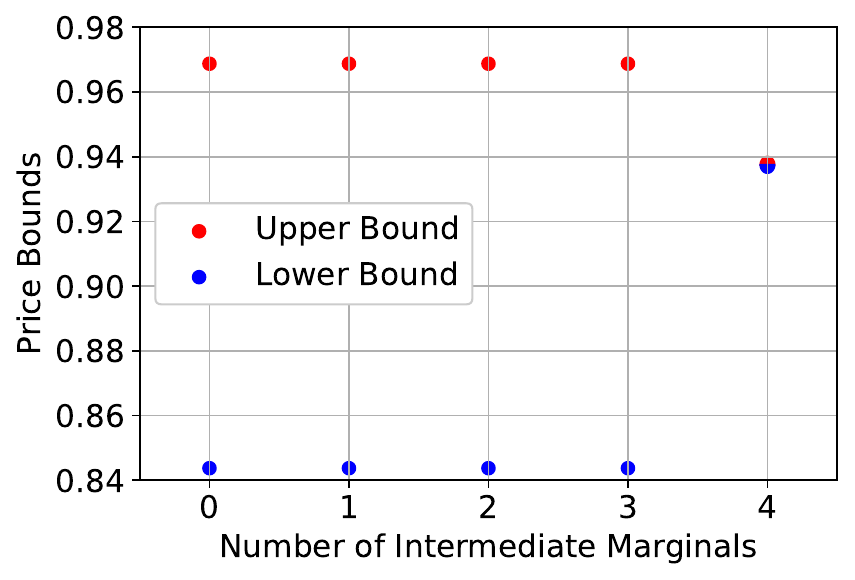} \label{fig_binomial_left}}
    \subfigure[Including intermediate marginals from the right.]{
    \includegraphics[width=0.45\textwidth]{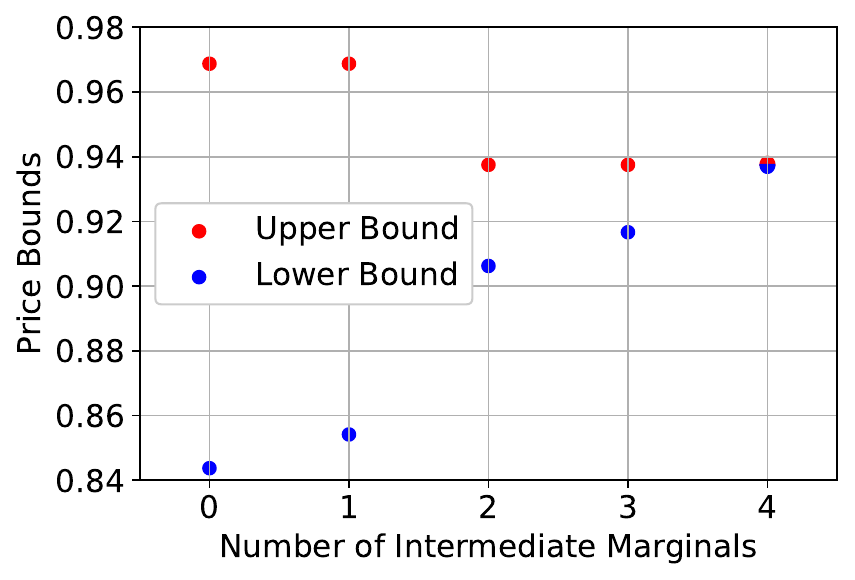}\label{fig_binomial_right}}
              \caption{The two plots show the improvement implied by the inclusion of intermediate marginals in the setting of Example~\ref{exa_binomial}. The left panel shows the improvement when starting the inclusion of intermediate marginals with $\mu_2$ and ending with $\mu_5$, whereas the right panel shows the improvement when including respecting information in reverse order, i.e., first including $\mu_5$ and eventually incorporating $\mu_2$.}\label{fig_information_left_right}
\end{figure}
\end{center}
Compare also Figure~\ref{fig_binomial_left} and Figure~\ref{fig_binomial_right} where we illustrate this behavior.
As discussed, e.g., in \cite{rothschild1978increasing}, $\mu_5$ involves more uncertainty than $\mu_2$ (which is ensured by the increasing convex order). Therefore, information on $\mu_5$ is more \emph{valuable} for improving the price bounds as it reduces more future uncertainty by imposing stronger restrictions on the set of possible joint distributions. This delivers a sound explanation why including marginals \emph{from the right} improves price bounds stronger than including marginals \emph{from the left}.

\end{exa}

The next example uses the setting from Section~\ref{sec_financial_application} where a finite amount of options are observed.

\begin{exa}\label{exa_artificial_calls}
We consider prices of call options that are indicated in Table~\ref{tbl_artificial} for strikes $(K_{i,j})_{j=0,\dots,6}=\{0,50,80,100,120,200,250\}$ for $i=1,2,3$. See also the left plot of Figure~\ref{fig_art_calls}. 
\begin{table}[h!]
\begin{center}{
\begin{tabular}{lccccccc} \toprule
             & $\Pi_{i,0}$ &$\Pi_{i,1}$ &$\Pi_{i,2}$ &$\Pi_{i,3}$ &$\Pi_{i,4}$ &$\Pi_{i,5}$ &$\Pi_{i,6}$ \\
\midrule
$i=1$ & 100 &50&23&6&3&0.2&0 \\
$i=2$ & 100 &53&24.8&6&5.2&2&0 \\
$i=3$ & 100 &57&34&20&8&2&0 \\
\end{tabular}}
\end{center}
\caption{The prices of the considered call options from Example~\ref{exa_artificial_calls}.}\label{tbl_artificial}
\end{table}
\begin{center}
\begin{figure}[h!]
	\subfigure[The prices of the call options.]{
    \includegraphics[width=0.48\textwidth]{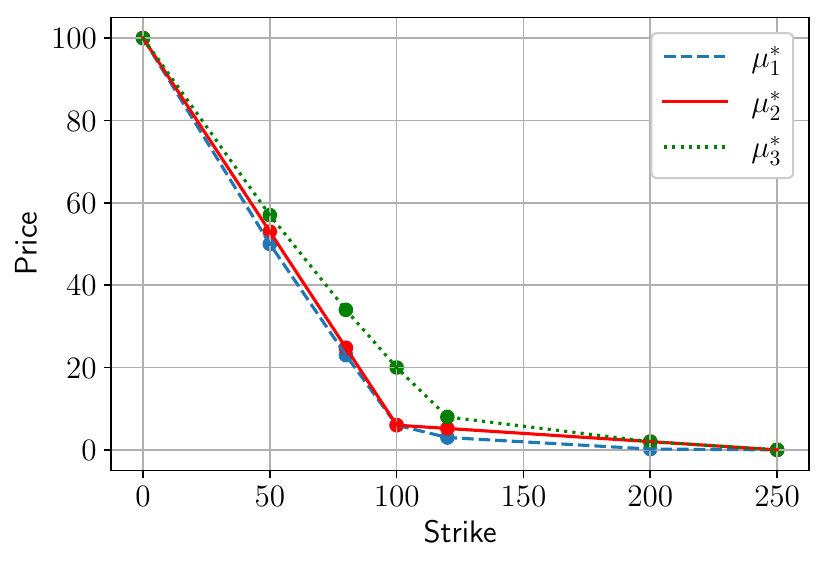} \label{fig_calls_artificial_Prices}}
    \subfigure[The implied marginals $\mu_i^*$.]{
    \includegraphics[width=0.48\textwidth]{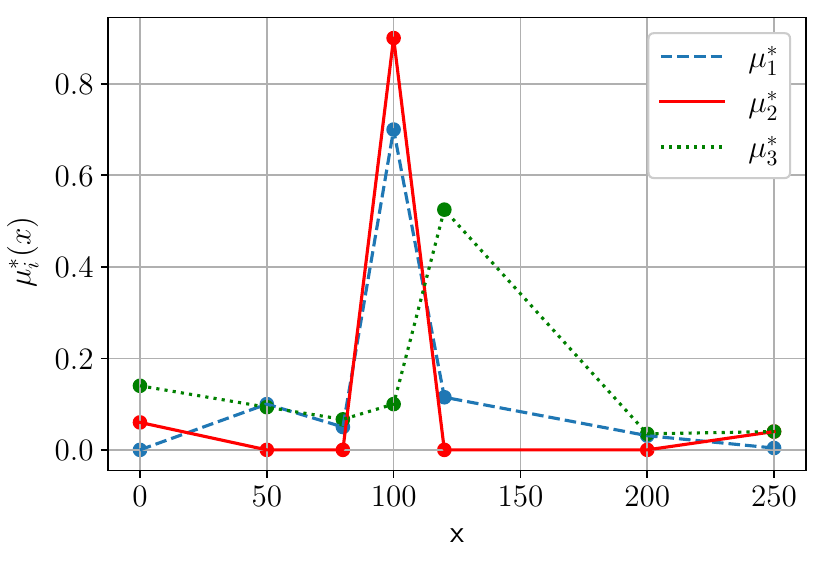}\label{fig_calls_artificial}}
              \caption{In the setting of Example~\ref{exa_artificial_calls} we illustrate the considered prices of the call options (a) as well as the implied marginal distributions $\mu_1^*,\mu_2^*,\mu_3^*$ (b).}\label{fig_art_calls}
\end{figure}
\end{center}
Note that, in particular, according to the discussion after Proposition~\ref{prop_mot_finite_options}, Assumption~\ref{asu_strikes} is fulfilled. 
We then derive the marginal distributions $\mu_i^*$, $i=1,2,3$, defined in Section~\ref{sec_financial_application}, and compute, by using Proposition~\ref{prop_mot_finite_options}, the improvement of the three-marginal MOT problem over the two-marginal MOT problem, i.e., the  difference between $\underline{P}_{1,2,3}(c)$ and $\underline{P}_{1,3}(c)$ for different payoff functions $c$. Moreover we study the improvement of the associated upper bounds $\overline{P}_{1,2,3}(c):=-\underline{P}_{1,2,3}(-c)$ and $\overline{P}_{1,3}(c):=-\underline{P}_{1,3}(-c)$. The improvements are indicated in Table~\ref{tbl_artificial_improve} and they reveal that including the intermediate marginal has a relatively strong effect on the price bounds. This can be well explained since the marginal distribution $\mu_2^*$ is constructed as an \emph{extreme} case which, in contrast to $\mu_1^*$ and $\mu_3^*$, only  possesses positive mass at $3$ atoms and is therefore rather restrictive with respect to possible joint distributions, obviously excluding joint distributions that were optimal in the two-marginal case, in line with Proposition~\ref{prop_intermediate_projection}~(vi).  Moreover, Table~\ref{tbl_artificial_improve} shows, that the degree of improvement highly depends on the considered payoff function.

\begin{table}[h!]
\begin{center}{
\begin{tabular}{lcccccc} \toprule
Payoff    $c(x_1,x_3)$        & $\underline{P}_{1,3}$ &$\underline{P}_{1,2,3}$ &$\overline{P}_{1,2,3}$ &$\overline{P}_{1,3}$ &$\frac{\underline{P}_{1,2,3}-\underline{P}_{1,3}}{\underline{P}_{1,3}}$ & $\frac{\overline{P}_{1,3}-\overline{P}_{1,2,3}}{\overline{P}_{1,3}}$\\
\midrule

$|x_3-x_1|$ & 28.13&31.63& 39.99&39.99&12.43$\%$& 0.0$\%$\\
$(\frac{1}{2}(x_1+x_3)-70)^+$ & 33.57& 33.68&  35.01&35.14
& 0.33 $\%$&0.38$\%$
 \\
$(\frac{1}{2}(x_1+x_3)-100)^+$ & 11.08&11.11& 12.83& 13.0&
0.32$\%$&1.28$\%$ \\
$(\frac{1}{2}(x_1+x_3)-130)^+$ &3.26 &  3.58&4.6& 4.75&
9.61$\%$&3.26$\%$
\\
\end{tabular}}
\end{center}
\caption{The improvement  of $\underline{P}_{1,2,3}(c)$ over $\underline{P}_{1,3}(c)$, and of $\overline{P}_{1,2,3}(c)$ over $\overline{P}_{1,3}(c)$ for different payoff functions $c$.}\label{tbl_artificial_improve}
\end{table}
\end{exa}

We indicate within the following example on which scale the improvement emerges in a more realistic situation.

\begin{exa}\label{exa_real}
We consider market prices\footnote{We use the mid-prices, i.e., the mean of observed bid and ask prices.} of call options written on the stock of \emph{Alphabet Inc.\,} (Ticker GOOG) and prices of call options written on the stock of \emph{Apple Inc.\,} (Ticker AAPL). The prices for these call options were observed on $11$ October $2022$. After receiving the prices from \emph{Yahoo Finance}, we apply the methodology from \cite{cohen2020detecting} that allows by slightly altering the prices (minimally with respect to the $\ell^1$ -norm) to obtain arbitrage-free prices which fulfil Assumption~\ref{asu_strikes}. This procedure also ensures that the associated marginals $\mu_i^*$, $i=1,2,3$ fulfil Assumption~\ref{asu_marginals}. We then compute, by applying Proposition~\ref{prop_mot_finite_options}, the improvement of $\underline{P}_{1,2,3}(c)$ over $\underline{P}_{1,3}(c)$ and of $\overline{P}_{1,2,3}(c)$ over $\overline{P}_{1,3}(c)$ for a forward start straddle with payoff $c(S_{t_1},S_{t_3})=|S_{t_3}-S_{t_1}|$ for different combinations of maturities $t_1<t_2<t_3$.

\begin{table}[h!]
\begin{center}{
\begin{tabular}{lcccc} \toprule
&\multicolumn{2}{c}{Alphabet Inc.  } &\multicolumn{2}{c}{Apple Inc.   } \\ 
\midrule
Maturities $(t_1,t_2,t_3)$            &$\frac{\underline{P}_{1,2,3}-\underline{P}_{1,3}}{\underline{P}_{1,3}}$ & $\frac{\overline{P}_{1,3}-\overline{P}_{1,2,3}}{\overline{P}_{1,3}}$ &$\frac{\underline{P}_{1,2,3}-\underline{P}_{1,3}}{\underline{P}_{1,3}}$ & $\frac{\overline{P}_{1,3}-\overline{P}_{1,2,3}}{\overline{P}_{1,3}}$\\
\midrule
(0.02, 0.1, 0.18) 	&0.47 &0 &0.83 &0.65 \\
(0.02, 0.1, 0.25) 	&0.26 &0 &0.20 &0.02 \\
(0.02, 0.1, 0.42)  &0.17 &0 &0.08 &0 \\
(0.02, 0.18, 0.25) 	&0.08 &0 &0.05 &0 \\
(0.02, 0.18, 0.42) 	&0.02 &0 &0.03 &0 \\
(0.02, 0.25, 0.42) 	&0.01 &0 &0.00 &0 \\
(0.1, 0.18, 0.25) 	&1.29 &0 &0.03 &0 \\
(0.1, 0.18, 0.42) 	&0.39 &0 &1.54 &0 \\
(0.1, 0.25, 0.42) 	&0.08 &0 &3.44 &0 \\
(0.18, 0.25, 0.42) 	&0.21 &0 &1.42 &0 \\
\end{tabular}}
\end{center}
\caption{The improvement (in $\%$) of $\underline{P}_{1,2,3}(c)$ over $\underline{P}_{1,3}(c)$ and of $\overline{P}_{1,2,3}(c)$ over $\overline{P}_{1,3}(c)$ for the payoff function $
c(x_1,x_3):=|x_3-x_1|.
$ The marginals are derived from call options written on Alphabet Inc. and Apple Inc., respectively. Note that the maturities are measured in years.}\label{tbl_calls_1}
\end{table}
The results are summarized in Table~\ref{tbl_calls_1} and they reveal that an improvement is in most cases observable, and the improvement tends to be slightly larger when more distant maturities are involved, which is in line with the discussion from Example~\ref{exa_binomial} and from {\cite{rothschild1978increasing}} stating that an increasing convex order\footnote{Note that the largest maturities correspond to marginals that are the largest with respect to the convex order.} comes with an increasing \emph{risk}, and therefore with more uncertainty which can be reduced by including additional intermediate marginals. Since our study does not take into account any interest rates nor dividend yields we however only considered rather short maturities. A larger improvement can be expected for more distant maturities. 

While the amount of resulting improvement of the bounds turns out to be on a rather small scale, we highlight again that including additional intermediate marginal information does not require to impose any additional assumptions nor to expensively collect  information. The approach simply uses the entire information available in the market. Therefore, the results imply, since intermediate marginals evidently may have a (small) impact on price bounds that it is strongly advisable to use as much price information about liquid market instruments as available.

\end{exa}

\section{Proofs}\label{sec_proofs}
In Section~\ref{sec_appendix_auxiliary} we report auxiliary results that are helpful to establish the proofs the main results which are provided in Section~\ref{sec_appendix_proofs_main}.
\subsection{Auxiliary Results}\label{sec_appendix_auxiliary}

The first lemma states a duality result which was initially proved in \cite{beiglbock2013model}. We refer also to \cite[Theorem 2.2.]{bartl2019robust} and  \cite[Theorem 2.1]{zaev2015monge}, where more general duality results are provided.
\begin{lem}[\cite{beiglbock2013model}, Theorem 1.1.]\label{lem_beiglboeck_dual}
Let $n \in \N$,  $c\in \mathcal{C}_{\operatorname{lin}}(\R^n)$ and assume that Assumption~\ref{asu_marginals} is fulfilled for $(\mu_i)_{i=1,\dots,n} \subset  \mathcal{P}(\R)$.
Then it holds
\begin{align}
&\sup_{u_i \in L^1(\mu_i),\Delta_i \in \mathcal{C}_b(\R^i)}\bigg\{\sum_{i=1}^n\E_{\mu_i}[u_i]~\bigg|~\sum_{i=1}^nu_i(x_i)+\sum_{i=1}^{n-1}\Delta_i(x_1,\dots,x_i)(x_{i+1}-x_i)\notag \\
&\hspace{6cm}\leq c(x_1,\dots,x_n)~\text{ for all } (x_1,\dots,x_n) \in \R^n\bigg\} \notag \\
&=\sup_{u_i \in \C_b(\R),\Delta_i \in \mathcal{C}_b(\R^i)}\bigg\{\sum_{i=1}^n\E_{\mu_i}[u_i]~\bigg|~\sum_{i=1}^nu_i(x_i)+\sum_{i=1}^{n-1}\Delta_i(x_1,\dots,x_i)(x_{i+1}-x_i)\notag \\
&\hspace{6cm}\leq c(x_1,\dots,x_n)~\text{ for all } (x_1,\dots,x_n) \in \R^n\bigg\} \notag \\
&=\inf_{\Q \in \mathcal{M}(\mu_1,\dots,\mu_n)}\E_\Q[c(S_{t_1},\dots,S_{t_n})]. \label{eq_rhs_duality}
\end{align}
Moreover, the infimum in \eqref{eq_rhs_duality} is attained by some measure $\Q \in \mathcal{M}(\mu_1,\dots,\mu_n)$.
\end{lem}

A proof of the following observation can be found in \cite[Corollary 2.1]{henry2017modelfree}. It states that the value of the  model-independent sub-hedging problem is not influenced when we additionally trade at an intermediate time.

\begin{lem}[\cite{henry2017modelfree}, Corollary 2.1]\label{lem_intermediate_trading}
Let $c\in \mathcal{C}_{\operatorname{lin}}(\R^2)$ and assume that Assumption~\ref{asu_marginals} is fulfilled for $\mu_1,\mu_3 \in \mathcal{P}(\R)$.
Then it holds
\begin{align*}
\lefteqn{\sup_{u_i \in L^1(\mu_i),\Delta_i \in \mathcal{C}_b(\R^i)}\bigg\{\E_{\mu_1}[u_1]+\E_{\mu_3}[u_3]~\big|~u_1(x_1)+u_3(x_3)+\Delta_1(x_1)(x_2-x_1)}\\
&\hspace{4cm}+\Delta_2(x_1,x_2)(x_3-x_2) \leq c(x_1,x_3)~\text{ for all } (x_1,x_2,x_3) \in \R^3\bigg\} \\
&= \lefteqn{\sup_{u_i \in \C_b(\R),\Delta_i \in \mathcal{C}_b(\R^i)}\bigg\{\E_{\mu_1}[u_1]+\E_{\mu_3}[u_3]~\big|~u_1(x_1)+u_3(x_3)+\Delta_1(x_1)(x_2-x_1)}\\
&\hspace{4cm}+\Delta_2(x_1,x_2)(x_3-x_2) \leq c(x_1,x_3)~\text{ for all } (x_1,x_2,x_3) \in \R^3\bigg\} \\
&=\inf_{\Q \in \mathcal{M}(\mu_1,\mu_3)}\E_\Q[c(S_{t_1},S_{t_3})].
\end{align*}
\end{lem}
Next, we establish the following assertion ensuring that the function $H_{(u_i),(\Delta_i)}$, defined in \eqref{eq_defn_H}, is integrable if it is bounded from below which is in particular the case if we require $H_{(u_i),(\Delta_i)} \geq 0$.

\begin{lem}\label{lem_contiinuity_H}
Let $c\in \mathcal{C}_{\operatorname{lin}}(\R^2)$. Let $u_1,u_3 \in \C_b(\R)$, $\Delta_1 \in C_b(\R), \Delta_2 \in \C_b(\R^2)$, and let $H_{(u_i),(\Delta_i)}$ be defined as in \eqref{eq_defn_H} and assume that Assumption~\ref{asu_marginals} is fulfilled for $\mu_1,\mu_2,\mu_3 \in \mathcal{P}(\R)$. If there exists some $u_2\in L^1(\mu_2)$ such that 
\begin{equation}\label{eq_condition_Lemma53}
\sum_{i=1}^3u_i(x_i)+{\Delta_1}(x_1)(x_2-x_1)+{\Delta_1}(x_1,x_2)(x_3-x_2) \leq c(x_1,x_3) \text{ for all } x_1,x_2,x_3 \in \R,
\end{equation}
then we have
\[
H_{(u_i),(\Delta_i)} \in L^1(\mu_2).
\]
\end{lem}

\begin{proof}[Proof of Lemma~\ref{lem_contiinuity_H}]
As a pointwise infimum of continuous functions, the function $H_{(u_i),(\Delta_i)}$ is upper semicontinuous and in particular measurable.
Further, note that condition \eqref{eq_condition_Lemma53} implies that 
\begin{equation}\label{eq_H_small_linear_growth_0}
H_{(u_i),(\Delta_i)}  \geq u_2 \in L^1(\mu_2).
\end{equation}
Since $u_1,u_3,\Delta_1 \in C_b(\R), \Delta_2 \in \C_b(\R^2)$, there exists some $C>0$ such that for all $x_2 \in \R$ we have
\begin{equation}\label{eq_H_small_linear_growth}
H_{(u_i),(\Delta_i)}(x_2) \leq \inf_{x_1,x_3 \in \R} \left\{C(1+|x_1|+|x_2|+|x_3|)\right\}= C(1+|x_2|).
\end{equation}
By Assumption~\ref{asu_marginals}, the first moment of $\mu_2$ exists. Hence \eqref{eq_H_small_linear_growth_0} and \eqref{eq_H_small_linear_growth} together show that $H_{(u_i),(\Delta_i)}$  is indeed $\mu_2$-integrable.
\end{proof}

In the setting of Section~\ref{sec_financial_application} we further establish the following assertion.
\begin{lem}\label{lem_financial_intermediate}
Let $c:\R^2 \rightarrow \R$. Then, we have
\begin{equation}\label{eq_finanical_intermediate_proposition}
\begin{aligned}
\underline{P}_{1,3}(c)=\sup_{d,\lambda_{i,j}, \Delta_0\in \R, \atop  \Delta_1 \in \mathcal{C}_b(\R),\Delta_2 \in \mathcal{C}_b(\R^2)} \bigg\{  d+\sum_{i=1,3}\sum_{j=0}^{m_i} \lambda_{i,j}\Pi_{i,j}~\bigg|~&\Psi^{1,2,3}_{d,(\lambda_{i,j}),(\Delta_i)}(x_1,x_2,x_3) \leq c(x_1,x_3) , \\&\hspace{-1cm}\text{ for all } x_i \in  (K_{i,j})_{j=0,\dots,m_i} \text{ for } i=1,2,3  \\
&\hspace{0.5cm}\text{ and }\lambda_{2,j} = 0 \text{ for } j=0,\dots,m_2 \bigg\}.
\end{aligned}
\end{equation}
\end{lem}
\begin{proof}[Proof of Lemma~\ref{lem_financial_intermediate}]
The inequality $\leq$ in \eqref{eq_finanical_intermediate_proposition} follows immediately. To show the other inequality, let $d,\lambda_{i,j}, \Delta_0 \in \R, \Delta_1 \in \mathcal{C}_b(\R),\Delta_2 \in \mathcal{C}_b(\R^2)$ such that $\Psi^{1,2,3}_{d,(\lambda_{i,j}),(\Delta_i)} \leq c$, and such that $\lambda_{2,j} = 0 \text{ for } j=0,\dots,m_2 $. Then, we define $\widetilde{\Delta}_0:= \Delta_0$ as well as $\R \ni x_1 \mapsto \widetilde{\Delta}_1(x_1):=\Delta_2(x_1,x_1)$. One directly sees that
\[
\Psi^{1,3}_{d,(\lambda_{i,j}),(\widetilde{\Delta}_i)} \leq c
\] 
which implies the remaining inequality.
\end{proof}
The following two lemmas are crucial to prove Proposition~\ref{prop_mot_finite_options}.
\begin{lem}\label{lem_finite_options}
Let Assumption~\ref{asu_strikes} hold true, let $i\in \{1,2,3\}$ and let $f \in \Clin(\R)$. Then, there exists some $u\in \mathcal{S}_i$ with
\[
{\mathcal{S}}_i :=\left\{u:\R \rightarrow \R~\middle|~u(x)=d+\sum_{j=0}^{m_i}\lambda_j (x-K_{i,j})^++\Delta_0(x-S_0) \text{ for some } d, \Delta_0 , \lambda_j \in \R \right\}
\]
such that it holds 
\[
f(x) = u(x) \text{ for all } x \in \{K_{i,0},\dots,K_{i,m_i}\},
\]
as well as
\[
\E_{\mu_i^*}[f]=\E_{\mu_i^*}[u].
\]
\end{lem}
\begin{proof}
Let $f \in \Clin(\R)$. Then pick $d, \Delta_0 \in \R$ such that 
\[
f(K_{i,0})=d+\Delta_0(K_{i,0}-S_0).
\]
We continue by choosing some $\lambda_0\in \R$ such that 
\[
f(K_{i,1})= d+ \lambda_0 (K_{i,1}-K_{i,0})+\Delta_0(K_{i,1}-S_0)
\]
and iteratively for all $j=1,\dots,m_i-1$ we pick some $\lambda_j \in \R$
such that 
\[
f(K_{i,j+1})= d+   {\lambda_j} (K_{i,j+1}-K_{i,j})+ \sum_{k=0}^{j-1} {\lambda_k} (K_{i,j+1}-K_{i,k})+\Delta_0(K_{i,j+1}-S_0).
\]
We choose an arbitrary value $\lambda_{m_i}\in \R$ and define the function
\[
\R \ni x \mapsto u(x):= d+  \sum_{k=0}^{m_i} {\lambda_k} (x-K_{i,k})^++\Delta_0(x-S_0)
\]
which fulfils, by construction, that $u = f $ on $ \{K_{i,0},\dots,K_{i,m_i}\}$ as well as
\[
\E_{\mu_i^*}[u]=\E_{\mu_i^*}[f]
\]
since $\mu_i^*$ is supported on $ \{K_{i,0},\dots,K_{i,m_i}\}$ .
\end{proof}

\begin{lem}\label{lem_xi_dual}
Let $\Xi \subseteq \R^n$ be closed, let $n \in \N$,  $c\in \mathcal{C}_{\operatorname{lin}}(\R^n)$ and assume that Assumption~\ref{asu_marginals} is fulfilled for $(\mu_i)_{i=1,\dots,n} \subset  \mathcal{P}(\R)$. Moreover, assume that $\{\Q \in \mathcal{M}(\mu_1,\dots,\mu_n), \Q(\Xi)=1\} \neq \emptyset$.
Then it holds
\begin{align}
&\sup_{u_i \in \Clin(\R),\Delta_i \in \mathcal{C}_b(\R^i)}\bigg\{\sum_{i=1}^n\E_{\mu_i}[u_i]~\bigg|~\sum_{i=1}^nu_i(x_i)+\sum_{i=1}^{n-1}\Delta_i(x_1,\dots,x_i)(x_{i+1}-x_i)\notag \\
&\hspace{6cm}\leq c(x_1,\dots,x_n)~\text{ for all } (x_1,\dots,x_n) \in \Xi\bigg\} \notag \\
&=\inf_{\Q \in \mathcal{M}(\mu_1,\dots,\mu_n)} \left\{\E_\Q[c(S_{t_1},\dots,S_{t_n})]~\middle|~\Q(\Xi)=1\right\}. \label{eq_rhs_duality}
\end{align}
\end{lem}
\begin{proof}
This follows, e.g.,  from \cite[Theorem 2.4~(b)]{neufeld2021model} when considering no dynamic option trading, i.e., $V= \emptyset$ in the notation of \cite{neufeld2021model}.
\end{proof}
\subsection{Proofs of the Main Results}\label{sec_appendix_proofs_main}
\begin{proof}[Proof of Lemma~\ref{lem_positivity_of_H}]
This follows directly by definition of $H_{(u_i),(\Delta_i)}$.
\end{proof}
\begin{proof}[Proof of Proposition~\ref{prop_intermediate_projection}]~\\
\framebox{$(i) \Leftrightarrow (ii)$}\\
This follows directly by Proposition~\ref{prop_degree_of_improvement}.\\
\framebox{$(i) \Rightarrow (iii)$}\\
Let $\varepsilon>0$. According to Lemma~\ref{lem_intermediate_trading} we find some $u_1^\varepsilon \in \C_b(\R)$, $u_3^\varepsilon \in \C_b(\R)$, $\Delta_1^\varepsilon \in \C_b(\R)$, $\Delta_2^\varepsilon\in \C_b(\R^2)$ with
\begin{equation}\label{eq_ineq_ui_itoiii_no1}
u_1^\varepsilon(x_1)+u_3^\varepsilon(x_3)+\Delta_1^\varepsilon(x_1)(x_2-x_1)+\Delta_2^\varepsilon(x_1,x_2)(x_3-x_2) \leq c(x_1,x_3)~\text{ for all } (x_1,x_2,x_3) \in \R^3
\end{equation}
such that
\begin{equation}\label{eq_ineq_ui_itoiii}
\inf_{\Q \in \mathcal{M}(\mu_1,\mu_2,\mu_3)}\E_\Q[c(S_{t_1},S_{t_3})] = \inf_{\Q \in \mathcal{M}(\mu_1,\mu_3)}\E_\Q[c(S_{t_1},S_{t_3})] < \E_{\mu_1}[u_1^\varepsilon]+\E_{\mu_3}[u_3^\varepsilon]+\varepsilon.
\end{equation}
Note that by Lemma~\ref{lem_positivity_of_H}~(i), inequality \eqref{eq_ineq_ui_itoiii_no1} is equivalent to $H_{(u_i^\varepsilon),(\Delta_i^\varepsilon)}\geq 0$.
By Lemma~\ref{lem_positivity_of_H}~(ii) we also have
\begin{equation}\label{eq_ineq_H_itoiii}
u_1^\varepsilon(x_1)+H_{(u_i^\varepsilon),(\Delta_i^\varepsilon)}(x_2)+u_3^\varepsilon(x_3)+\Delta_1^\varepsilon(x_1)(x_2-x_1)+\Delta_2^\varepsilon(x_1,x_2)(x_3-x_2) \leq c(x_1,x_3)~
\end{equation}
for all $(x_1,x_2,x_3) \in \R^3$.
Then, by Lemma~\ref{lem_contiinuity_H} it follows $H_{(u_i^\varepsilon),(\Delta_i^\varepsilon)}\in L^1(\mu_2)$. We integrate both sides of the inequality \eqref{eq_ineq_H_itoiii} with respect to some $\Q\in \mathcal{M}(\mu_1,\mu_2,\mu_3)$ and obtain
\begin{equation}\label{eq_ineq_H_eps_proof_i_iii}
\E_{\mu_1}[u_1^\varepsilon]+\E_{\mu_2}[H_{(u_i^\varepsilon),(\Delta_i^\varepsilon)}]+\E_{\mu_3}[u_3^\varepsilon]\leq \E_\Q[c(S_{t_1},S_{t_3})].
\end{equation}
As $\Q\in \mathcal{M}(\mu_1,\mu_2,\mu_3)$ was arbitrary, the inequality  \eqref{eq_ineq_H_eps_proof_i_iii} implies together with \eqref{eq_ineq_ui_itoiii} that
\[
\E_{\mu_1}[u_1^\varepsilon]+\E_{\mu_2}[H_{(u_i^\varepsilon),(\Delta_i^\varepsilon)}]+\E_{\mu_3}[u_3^\varepsilon]\leq \inf_{\Q \in \mathcal{M}(\mu_1,\mu_2,\mu_3)}\E_\Q[c(S_{t_1},S_{t_3})] <\E_{\mu_1}[u_1^\varepsilon]+\E_{\mu_3}[u_3^\varepsilon]+\varepsilon,
\]
and thus, as  $H_{(u_i^\varepsilon),(\Delta_i^\varepsilon)}\geq 0$, also
\[
0\leq \E_{\mu_2}[H_{(u_i^\varepsilon),(\Delta_i^\varepsilon)}]< \varepsilon.
\]
Hence, we have
\[
\left|\E_{\mu_1}[u_1^\varepsilon]+\E_{\mu_2}[H_{(u_i^\varepsilon),(\Delta_i^\varepsilon)}]+\E_{\mu_3}[u_3^\varepsilon]-\inf_{\Q \in \mathcal{M}(\mu_1,\mu_2,\mu_3)}\E_\Q[c(S_{t_1},S_{t_3})]\right|<\varepsilon.
\]
Thus, (iii) follows.\\
\framebox{(iii) $\Rightarrow$ (i)}\\
Let (iii) be true, and assume that 
\begin{equation}\label{asu_greater_eps}
\inf_{\Q \in \mathcal{M}(\mu_1,\mu_2,\mu_3)}\E_\Q[c(S_{t_1},S_{t_3})]-\inf_{\Q \in \mathcal{M}(\mu_1,\mu_3)}\E_\Q[c(S_{t_1},S_{t_3})]=2\varepsilon \text{ for some }\varepsilon>0.
\end{equation}
Let $u_1,u_3,\Delta_1 \in \C_b(\R),\Delta_2 \in \C_b(\R^2)$ be as in the assertion of (iii) for $\varepsilon$, then we obtain by Lemma~\ref{lem_positivity_of_H}~(i), due to $H_{(u_i),(\Delta_i)}\geq 0$, that 
\[
u_1(x_1)+u_3(x_3)+\Delta_1(x_1)(x_2-x_1)+\Delta_2(x_1,x_2)(x_3-x_2) \leq c(x_1,x_3) \text{ for all }x_1,x_2,x_3 \in \R,
\]
and thus with $x_2=x_3$
\[
u_1(x_1)+u_3(x_3)+\Delta_1(x_1)(x_3-x_1)\leq c(x_1,x_3) \text{ for all }x_1,x_3 \in \R,
\]
implying also that 
\[
\E_{\mu_1}[u_1]+\E_{\mu_3}[u_3] \leq \inf_{\Q \in \mathcal{M}(\mu_1,\mu_3)}\E_\Q[c(S_{t_1},S_{t_3})].
\]
We then obtain  a contradiction by
\begin{align*}
2\varepsilon &= \inf_{\Q \in \mathcal{M}(\mu_1,\mu_2,\mu_3)}\E_\Q[c(S_{t_1},S_{t_3})]-\inf_{\Q \in \mathcal{M}(\mu_1,\mu_3)}\E_\Q[c(S_{t_1},S_{t_3})]\\
&\leq \inf_{\Q \in \mathcal{M}(\mu_1,\mu_2,\mu_3)}\E_\Q[c(S_{t_1},S_{t_3})]-(\E_{\mu_1}[u_1]+\E_{\mu_3}[u_3])\\
&\leq \left|\inf_{\Q \in \mathcal{M}(\mu_1,\mu_2,\mu_3)}\E_\Q[c(S_{t_1},S_{t_3})]-(\E_{\mu_1}[u_1]+\E_{\mu_2}[H_{(u_i^\varepsilon),(\Delta_i^\varepsilon)}]+\E_{\mu_3}[u_3])\right|+\left|\E_{\mu_2}[H_{(u_i^\varepsilon),(\Delta_i^\varepsilon)}]\right|<2\varepsilon,
\end{align*}
where the last strict inequality follows from \eqref{eq_ineq_requirement_lemma_ii}  and \eqref{eq_ineq_requirement_H_smaller_eps}.
And thus there exists no $\varepsilon>0$ such that \eqref{asu_greater_eps} is true.\\
\framebox{(iii) $\Rightarrow$ (iv)}\\
Let $\varepsilon>0$. According to (iii), there exists some $u_1,u_3, \Delta_1 \in \C_b(\R)$, $\Delta_2 \in \C_b(\R^2)$ such that
\begin{equation}\label{eq_23_eq_1}
\left|\E_{\mu_1}[u_1]+\E_{\mu_2}[H_{(u_i),(\Delta_i)}]+\E_{\mu_3}[u_3]-\inf_{\Q \in \mathcal{M}(\mu_1,\mu_2,\mu_3)}\E_\Q[c(S_{t_1},S_{t_3})]\right|<\frac{\varepsilon}{2}
\end{equation}
and such that
\begin{equation}\label{eq_23_eq_2}
0\leq \E_{\mu_2}[H_{(u_i),(\Delta_i)}] < \frac{\varepsilon}{2}.
\end{equation}
Moreover, we have $H_{(u_i),(\Delta_i)}(x_2)\geq 0$ for all $x_2 \in \R$ which implies by Lemma~\ref{lem_positivity_of_H}~(i)
\[
u_1(x_1)+ u_3(x_3)+\Delta_1(x_1)(x_2-x_1) + \Delta_2(x_1,x_2)(x_3-x_2) \leq c(x_1,x_3) \text{ for all } x_1,x_2, x_3 \in \R,
\]
and thus through setting $x_2 = x_3$ also 
\[
u_1(x_1)+ u_3(x_3)+\Delta_1(x_1)(x_3-x_1) \leq c(x_1,x_3) \text{ for all } x_1, x_3 \in \R.
\]
Moreover, by \eqref{eq_23_eq_1} and by \eqref{eq_23_eq_2} we have 
\begin{align*}
&\left|\E_{\mu_1}[u_1]+\E_{\mu_3}[u_3]-\inf_{\Q \in \mathcal{M}(\mu_1,\mu_2,\mu_3)}\E_\Q[c(S_{t_1},S_{t_3})]\right|\\
\leq &\left|\E_{\mu_1}[u_1]+\E_{\mu_2}[H_{(u_i),(\Delta_i)}]+\E_{\mu_3}[u_3]-\inf_{\Q \in \mathcal{M}(\mu_1,\mu_2,\mu_3)}\E_\Q[c(S_{t_1},S_{t_3})]\right|+\left|\E_{\mu_2}[H_{(u_i),(\Delta_i)}]\right| < \varepsilon.
\end{align*}
\framebox{(iv) $\Rightarrow$ (iii)}\\
Let $\varepsilon>0$ and let $u_1,u_3,\Delta_1 \in \C_b(\R)$ such that \eqref{eq_ineq_requirement_lemma_iii} and \eqref{eq_ineq_requirement_lemma_iii_eps} hold true. Then we set $\Delta_2 :\equiv \Delta_1$ and obtain $H_{(u_i),(\Delta_i)}(x_2)\geq 0$ as a consequence of \eqref{eq_ineq_requirement_lemma_iii}, see {Lemma~\ref{lem_positivity_of_H}~(i)}. We then have by Lemma~\ref{lem_positivity_of_H}~(ii) that 
\[
u_1(x_1)+H_{(u_i),(\Delta_i)}(x_2)+u_3(x_3)+\Delta_1(x_1)(x_3-x_1) \leq c(x_1,x_3) \text{ for all } x_1,x_2,x_3 \in \R.
\]
In particular, Lemma~\ref{lem_contiinuity_H} is applicable showing that $H_{(u_i),(\Delta_i)}\in L^1(\mu_2)$, which implies
\[
\E_{\mu_1}[u_1]+ \E_{\mu_3}[u_3] \leq \E_{\mu_1}[u_1]+\E_{\mu_2}[H_{(u_i),(\Delta_i)}]+\E_{\mu_3}[u_3]\leq \inf_{\Q \in \mathcal{M}(\mu_1,\mu_2,\mu_3)}\E_\Q[c(S_{t_1},S_{t_3})],
\]
and thus \eqref{eq_ineq_requirement_lemma_ii} follows with \eqref{eq_ineq_requirement_lemma_iii_eps}.\\
\framebox{(i) $\Rightarrow$ (v)}\\
Let $\varepsilon>0$ and pick  $u_1 \in \mathcal{C}_b(\R), u_2 \in \mathcal{C}_b(\R), u_3 \in \mathcal{C}_b(\R)$, $\Delta_1 \in \mathcal{C}_b(\R)$, $\Delta_2 \in \mathcal{C}_b(\R^2)$ that satisfy \eqref{eq_ineq_3marginals}. Then, by (i) and by Lemma~\ref{lem_beiglboeck_dual} we have that 
\begin{align}
\E_{\mu_1}[u_1]+ \E_{\mu_2}[u_2]+\E_{\mu_3}[u_3]&\leq \inf_{\Q \in \mathcal{M}(\mu_1,\mu_2,\mu_3)}\E_\Q[c(S_{t_1},S_{t_3})] = \inf_{\Q \in \mathcal{M}(\mu_1,\mu_3)}\E_\Q[c(S_{t_1},S_{t_3})\notag] \\
&\hspace{-4cm}= \sup_{v_i \in L^1(\mu_i),\widetilde{\Delta}_1 \in \mathcal{C}_b(\R)}\bigg\{\E_{\mu_1}[v_1]+\E_{\mu_3}[u_3]~\big|~v_1(x_1)+v_3(x_3)+\widetilde{\Delta}_1(x_1)(x_3-x_1) \label{eq_inf_sup_i_to_v}\\
&\hspace{3.5cm}\leq c(x_1,x_3)~\text{ for all } (x_1,x_3) \in \R^3\bigg\}. \notag
\end{align}
And thus, if $\E_{\mu_1}[u_1]+ \E_{\mu_2}[u_2]+\E_{\mu_3}[u_3]< \inf_{\Q \in \mathcal{M}(\mu_1,\mu_2,\mu_3)}\E_\Q[c(S_{t_1},S_{t_3})]$, then by \eqref{eq_inf_sup_i_to_v} we can find $v_1,v_3 \in \C_b(\R),\widetilde{\Delta}_1 \in \mathcal{C}_b(\R)$ fulfilling \eqref{eq_ineq_2marginals} with \[
\E_{\mu_1}[u_1]+ \E_{\mu_2}[u_2]+\E_{\mu_3}[u_3] = \E_{\mu_1}[v_1]+\E_{\mu_3}[v_3].
\]
In case $\E_{\mu_1}[u_1]+ \E_{\mu_2}[u_2]+\E_{\mu_3}[u_3]= \inf_{\Q \in \mathcal{M}(\mu_1,\mu_2,\mu_3)}\E_\Q[c(S_{t_1},S_{t_3})]$, by \eqref{eq_inf_sup_i_to_v} we can find $v_1,v_3 \in\C_b(\R),\widetilde{\Delta}_1 \in \mathcal{C}_b(\R)$ fulfilling \eqref{eq_ineq_2marginals} and 
\[
\left|\E_{\mu_1}[v_1]+\E_{\mu_3}[v_3]-\inf_{\Q \in \mathcal{M}(\mu_1,\mu_3)}\E_\Q[c(S_{t_1},S_{t_3})]\right|< \varepsilon,
\]
implying the claim.\\
\framebox{(v) $\Rightarrow$ (iv)}\\
Let $\varepsilon>0$ and let $u_1 \in \mathcal{C}_b(\R), u_2 \in \mathcal{C}_b(\R), u_3 \in \mathcal{C}_b(\R)$, $\Delta_1 \in \mathcal{C}_b(\R)$, $\Delta_2 \in \mathcal{C}_b(\R^2)$ such that \eqref{eq_ineq_3marginals} holds true and such that 
\[
\left|\E_{\mu_1}[u_1]+ \E_{\mu_2}[u_2]+\E_{\mu_3}[u_3]-\inf_{\Q \in \mathcal{M}(\mu_1,\mu_2,\mu_3)}\E_\Q[c(S_{t_1},S_{t_3})]\right|< \frac{\varepsilon}{2}.
\]
This choice is possible due to Lemma~\ref{lem_beiglboeck_dual}.
Then according to (v), there exist $v_1,v_3 \in \mathcal{C}_b(\R),\widetilde{\Delta}_1 \in \mathcal{C}_b(\R)$ fulfilling \eqref{eq_ineq_2marginals} such that 
\[
\bigg|\E_{\mu_1}[v_1]+\E_{\mu_3}[v_3]-\left(\E_{\mu_1}[u_1]+\E_{\mu_2}[u_2]+\E_{\mu_3}[u_3]\right)\bigg|<\frac{\varepsilon}{2}.
\]
and thus
\[
\bigg|\E_{\mu_1}[v_1]+\E_{\mu_3}[v_3]-\inf_{\Q \in \mathcal{M}(\mu_1,\mu_2,\mu_3)}\E_\Q[c(S_{t_1},S_{t_3})]\bigg|<{\varepsilon}.
\]
\framebox{(i) $\Rightarrow$ (vi)}\\
Let (i) hold true. Then, we have
\begin{equation}\label{eq_14_eq_1}
\inf_{\Q \in \mathcal{M}(\mu_1,\mu_2,\mu_3)}\E_\Q[c(S_{t_1},S_{t_3})] = \inf_{\Q \in \mathcal{M}(\mu_1,\mu_3)}\E_{\Q}[c(S_{t_1},S_{t_3})].
\end{equation}
According to Lemma~\ref{lem_beiglboeck_dual} there exists some measure $\Q^*\in \mathcal{M}(\mu_1,\mu_2,\mu_3)$ such that 
\[
\E_{\Q^*}[c(S_{t_1},S_{t_3})]=\inf_{\Q \in \mathcal{M}(\mu_1,\mu_2,\mu_3)}\E_\Q[c(S_{t_1},S_{t_3})].
\]
Hence, with \eqref{eq_14_eq_1}, we obtain 
\[
\E_{\Q^*}[c(S_{t_1},S_{t_3})]=\E_{\pi(\Q^*)}[c(S_{t_1},S_{t_3})]= \inf_{\Q \in \mathcal{M}(\mu_1,\mu_3)}\E_{\Q}[c(S_{t_1},S_{t_3})],
\]
and thus
$\pi(\Q^*)\in \mathcal{Q}_c^*(\mu_1,\mu_3)$.
\\
\framebox{(vi) $\Rightarrow$ (i)}\\
Let $\Q^* \in \mathcal{M}(\mu_1,\mu_2,\mu_3)$ such that $\pi(\Q^*) \in \mathcal{Q}_c^*(\mu_1,\mu_3)$. Then, we have that
\begin{equation}\label{eq_14_eq_2}
\begin{aligned}
\inf_{\Q \in \mathcal{M}(\mu_1,\mu_2,\mu_3)}\E_\Q[c(S_{t_1},S_{t_3})]  &\leq \E_{\Q^*}[c(S_{t_1},S_{t_3})]\\
&=\E_{\pi(\Q^*)}[c(S_{t_1},S_{t_3})] \\
&= \inf_{\Q \in \mathcal{M}(\mu_1,\mu_3)}\E_{\Q}[c(S_{t_1},S_{t_3})] \\
&\leq \inf_{\Q \in \mathcal{M}(\mu_1,\mu_2,\mu_3)}\E_\Q[c(S_{t_1},S_{t_3})].
\end{aligned}
\end{equation}
\framebox{(vi) $\Rightarrow$ (vii)}\\
Let $\Q^* \in \mathcal{M}(\mu_1,\mu_2,\mu_3)$ such that $\pi(\Q^*) \in \mathcal{Q}_c^*(\mu_1,\mu_3)$. Then, we have
\begin{equation}\label{eq_iv_v_identity}
\E_{\Q^*}[c(S_{t_1},S_{t_3})]=\E_{\pi(\Q^*)}[c(S_{t_1},S_{t_3})]=\inf_{\Q \in \mathcal{M}(\mu_1,\mu_3)}\E_{\Q}[c(S_{t_1},S_{t_3})].
\end{equation}
\framebox{(vii) $\Rightarrow$ (vi)}\\
Let $\Q^* \in \mathcal{M}(\mu_1,\mu_2,\mu_3)$ such that $\E_{\Q^*}[c(S_{t_1},S_{t_3})]=\inf_{\Q \in \mathcal{M}(\mu_1,\mu_3)}\E_{\Q}[c(S_{t_1},S_{t_3})]$. Then, we have $\E_{\Q^*}[c(S_{t_1},S_{t_3})]=\E_{\pi(\Q^*)}[c(S_{t_1},S_{t_3})]$, and therefore with \eqref{eq_iv_v_identity} it follows $\pi(\Q^*) \in \mathcal{Q}_c^*(\mu_1,\mu_3)$.\\
\framebox{(vi) $\Leftrightarrow$ (viii)}\\
This follows directly when disintegrating $\Q\in \mathcal{M}(\mu_1,\mu_2,\mu_3)$ from the disintegration theorem for probability measures, see, e.g. \cite[Theorem 1]{chang1997conditioning} or \cite[Theorem 5.3.1]{ambrosio2005gradient}.
\end{proof}
{
\begin{proof}[Proof of Proposition~\ref{prop_characterization_mu_2}]{~}
\begin{itemize}
\item[(i)]
Before proving the assertion we note that, according to  \cite[Definition 6.8~(iv) and Theorem 6.9]{villani2009optimal}, for $d\in \N$,
the convergence $\PP^{(n)} \rightarrow \PP$   of a sequence $\left(\PP^{(n)}\right)_{n\in \N} \subseteq \mathcal{P}(\R^d)$ to some limit $\PP \in \mathcal{P}(\R^d)$ in the Wasserstein $1$-topology is equivalent to
\begin{equation}\label{eq_wassersteinspace_convergence}
\lim_{n\rightarrow\infty}\int_{\R^d} f(x)~\PP^{(n)}( \D  x) =\int_{\R^d} f(x)~\PP( \D  x) \text{ for all } f \in \Clin(\R^d).
\end{equation}

Let $\left(\mu_2^{(n)}\right)_{n \in \N} \subseteq \mathcal{I}$ such that $\mu_2^{(n)} \rightarrow \mu_2 \in \mathcal{P}(\R)$  as $n \rightarrow \infty$ w.r.t.\,the Wasserstein $1$-distance. Then, we observe $\mu_1 \preceq \mu_2 \preceq \mu_3$ since the Wasserstein-convergence implies by \eqref{eq_wassersteinspace_convergence} that for all $t \in \R$ we have 
 $$\lim_{n \rightarrow \infty}\int_\R \max\{x-t,0\}~ \mu_2^{(n)}(\D x) = \int_\R \max\{x-t,0\}~ \mu_2(\D x) 
$$
 which characterizes the convex order by, e.g.,  \cite[Theorem 2.1]{alfonsi2019sampling}. 

According to Proposition~\ref{prop_intermediate_projection}~(vi), for all $n \in \N$ there exists some $\Q^{(n)} \in \mathcal{M}(\mu_1,\mu_2^{(n)},\mu_2)$ such that $\pi(\Q^{(n)}) \in \mathcal{Q}_c^*(\mu_1,\mu_3)$. By \cite[Proposition 2.2.]{neufeld2021stability}, the set $\mathcal{Q}_c^*(\mu_1,\mu_3)$ is compact. Hence, there exists a subsequence $\left(\pi(\Q^{(n_k)})\right)_{k \in \N}$ such that $\pi(\Q^{(n_k)}) \rightarrow \Q^{*} \in \mathcal{Q}_c^*(\mu_1,\mu_3)$ as $k \rightarrow \infty$ w.r.t.\,the Wasserstein $1$-distance.

Note that, since $\mu_2^{(n_k)} \rightarrow \mu_2$ as $k \rightarrow \infty$ w.r.t.\,the Wasserstein $1$-distance and hence weakly, we have by Prokhorov's Theorem that $\left\{ \mu_2^{(n_k)},k \in \N \right\}$ is tight. 
We now denote by 
\[
\widetilde{\Pi}:= \Pi\left(\mu_1,\left\{ \mu_2^{(n_k)},k \in \N \right\}, \mu_2\right)
\]
the set of probability measures on $\R^3$ with first marginal $\mu_1$, third marginal $\mu_3$ and second marginal in $\left\{ \mu_2^{(n_k)},k \in \N \right\}$. Then, since $\left\{ \mu_2^{(n_k)},k \in \N \right\}$ is tight, for all $\varepsilon>0$ there exists some compact set $K_{\varepsilon,2}\subset \R$ such that for all $k \in \N$ we have $\mu_2^{(n_k)}(\R \backslash K_{\varepsilon,2} )<
\varepsilon$. Moreover, since $\mu_1$ and $\mu_3$ are probability measures, we can also find compact sets $K_{\varepsilon,1}\subset \R,K_{\varepsilon,3}\subset \R$ such that 
\[
\mu_1(\R \backslash K_{\varepsilon,1}) <\varepsilon,\qquad \mu_3(\R \backslash K_{\varepsilon,3}) <\varepsilon.
\]
Then, for all $\pi \in \widetilde{\Pi}$ we have 
\[
\pi \left(\R^3 \backslash (K_{\varepsilon,1}\times K_{\varepsilon,2} \times K_{\varepsilon,3})\right) \leq \mu_1(\R \backslash K_{\varepsilon,1})+\sup_{k \in \N} \mu_2^{(n_k)}(\R \backslash K_{\varepsilon,2})+\mu_3(\R \backslash K_{\varepsilon,3})< 3 \varepsilon.
\]
Hence $\widetilde{\Pi}$ is tight. Thus, according to Prokhorov's theorem there exists a subsequence $\left(\Q^{({n_k}_\ell)}\right)_{\ell \in \N}$ such that $\Q^{({n_k}_\ell)} \rightarrow \Q^{**} \in \mathcal{P}(\R^3)$ weakly as $\ell \rightarrow \infty$. We now show that $\Q^{**} \in \mathcal{M}(\mu_1,\mu_2,\mu_3)$. To this end, recall that by definition $\R^3 \ni (x_1,x_2,x_3) \mapsto S_{t_i}(x_1,x_2,x_3)= x_i \in \R$ for $i=1,2,3$. Then, we have by the continuous mapping theorem, that $\Q^{({n_k}_\ell)}\circ S_{t_i}^{-1} \rightarrow \Q^{**} \circ S_{t_i}^{-1}$ weakly as $\ell \rightarrow \infty$ and  $\Q^{({n_k}_\ell)}\circ S_{t_i}^{-1} = \mu_i^{({n_k}_\ell)} \rightarrow \mu_i$ weakly as $\ell \rightarrow \infty$ for $i=1,2,3$. Therefore, $\Q^{**}$ possesses marginals $\mu_1,\mu_2,\mu_3$. Further, as $\mu_2^{({n_k}_\ell)} \rightarrow \mu_2$ w.r.t.\,the Wasserstein $1$-distance as $\ell \rightarrow \infty$, we have by \eqref{eq_wassersteinspace_convergence} that 
\begin{align*}
&\lim_{\ell \rightarrow \infty} \int_{\R^3} \left( |x_1|+|x_2|+|x_3| \right)~ \Q^{({n_k}_\ell)}( \D  x_1, \D  x_2, \D  x_3) \\ 
&=\int_{\R} |x_1| ~\mu_1(\D x_1)+\lim_{\ell \rightarrow \infty} \int_\R |x_2| ~\mu_2^{({n_k}_\ell)}(\D x_2)+\int_{\R}|x_3| ~\mu_3(\D x_3) \\
&= \int_{\R} |x_1|~\mu_1(\D x_1)+\int_\R |x_2| ~\mu_2(\D x_2)+\int_{\R}|x_3| ~\mu_1(\D x_1)\\&= \int_{\R^3} \left( |x_1|+|x_2|+|x_3|\right) ~\Q^{**}( \D  x_1, \D  x_2, \D  x_3).
\end{align*}
Now,  \cite[Definition 6.8~(i)]{villani2009optimal} implies together with the weak convergence that $\Q^{({n_k}_\ell)} \rightarrow \Q^{**}$ w.r.t.\,the Wasserstein $1$-distance as $\ell \rightarrow \infty$.
Let $\Delta_1 \in \C_b(\R), \Delta_2 \in \C_b(\R^2)$ and note that both maps
\[
\R^3 \ni (x_1,x_2,x_3) \mapsto \Delta_1(x_1)(x_2-x_1),
\]
and
\[
\R^3 \ni (x_1,x_2,x_3) \mapsto \Delta_2(x_1,x_2)(x_3-x_2),
\]
are contained in $\Clin(\R^3)$,
which implies by \eqref{eq_wassersteinspace_convergence} that
\[
0= \lim_{\ell \rightarrow \infty} \int_{\R^3} \Delta_1(x_1)(x_2-x_1)~ \Q^{({n_k}_\ell)}( \D  x_1, \D  x_2, \D  x_3) = \int_{\R^3} \Delta_1(x_1)(x_2-x_1) ~\Q^{**}( \D  x_1, \D  x_2, \D  x_3),
\]
and 
\[
0= \lim_{\ell \rightarrow \infty} \int_{\R^3} \Delta_2(x_1,x_2)(x_3-x_2) ~\Q^{({n_k}_\ell)}( \D  x_1, \D  x_2, \D  x_3) = \int_{\R^3} \Delta_2(x_1,x_2)(x_3-x_2)~ \Q^{**}( \D  x_1, \D  x_2, \D  x_3). 
\]
This means we have shown that $\Q^{**} \in \mathcal{M}(\mu_1,\mu_2,\mu_3)$.

Then, it follows, since $c \in \Clin(\R^3)$, by the characterization in \eqref{eq_wassersteinspace_convergence} and with the Wasserstein $1$-convergences $\lim_{\ell \rightarrow \infty}\Q^{({n_k}_\ell)} = \Q^{**}$ and $\lim_{\ell \rightarrow \infty}\pi\left( \Q^{({n_k}_\ell)}\right)= \Q^{*}\in  \mathcal{Q}_c^*(\mu_1,\mu_3)$ that
\begin{align*}
\inf_{\Q \in \mathcal{M}(\mu_1,\mu_3)} \E_{\Q}[c(S_{t_1},S_{t_3})]&=\E_{\Q^*}[c(S_{t_1},S_{t_3})] \\
&= \lim_{\ell \rightarrow  \infty} \E_{\pi \left(\Q^{({n_k}_\ell)}\right)}[c(S_{t_1},S_{t_3})]\\
&= \lim_{\ell \rightarrow  \infty} \E_{\Q^{({n_k}_\ell)}}[c(S_{t_1},S_{t_3})]\\
&= \E_{\Q^{**}}[c(S_{t_1},S_{t_3})]= \E_{\pi \left(\Q^{**}\right)}[c(S_{t_1},S_{t_3})].
\end{align*}
This means $\pi \left(\Q^{**}\right) \in \mathcal{Q}_c^*(\mu_1,\mu_3)$ where $\Q^{**} \in \mathcal{M}(\mu_1,\mu_2,\mu_3)$ and hence, by Proposition~\ref{prop_intermediate_projection}~(iv), there is no improvement induced by the intermediate marginal $\mu_2$, i.e., $\mu_2 \in \mathcal{I}$.
\item[(ii)]
Let $\nu, \nu' \in \mathcal{I}$ and let $\underline{\Q} \in \mathcal{M}(\mu_1,\nu,\mu_3) , \underline{\Q}' \in \mathcal{M}(\mu_1,\nu',\mu_3)$ be minimizers of the associated MOT-problems, i.e.,
\begin{align*}
&\E_{\underline{\Q}}[c(S_{t_1},S_{t_3})] = \inf_{\Q \in \mathcal{M}(\mu_1,\nu,\mu_3)}\E_{\Q}[c(S_{t_1},S_{t_3})]= \inf_{\Q \in \mathcal{M}(\mu_1,\mu_3)}\E_{\Q}[c(S_{t_1},S_{t_3})], ~\\
&\E_{\underline{\Q}'}[c(S_{t_1},S_{t_3})] = \inf_{\Q \in \mathcal{M}(\mu_1,\nu',\mu_3)}\E_{\Q}[c(S_{t_1},S_{t_3})]=\inf_{\Q \in \mathcal{M}(\mu_1,\mu_3)}\E_{\Q}[c(S_{t_1},S_{t_3})].
\end{align*}
Let $\lambda \in(0,1)$ and define   $\mu_2 := \lambda \nu + (1-\lambda) \nu' \in \mathcal{P}(\R)$ and $\Q^*:= \lambda \underline{\Q} + (1-\lambda)\underline{\Q}' \in \mathcal{P}(\R^3)$. Then, we have $\Q^*\in \mathcal{M}(\mu_1,\mu_2, \mu_3)$ and hence
\begin{align*}
\inf_{\Q \in \mathcal{M}(\mu_1,\mu_2,\mu_3)}\E_{\Q}[c(S_{t_1},S_{t_3})] &\leq \E_{\Q^*}[c(S_{t_1},S_{t_3})] \\
&= \lambda \E_{\underline{\Q}}[c(S_{t_1},S_{t_3})] +(1-\lambda)\E_{\underline{\Q}'} [c(S_{t_1},S_{t_3})]  \\
&=  \lambda\inf_{\Q \in \mathcal{M}(\mu_1,\mu_3)} \E_{\Q}[c(S_{t_1},S_{t_3})] +(1-\lambda)\inf_{\Q \in \mathcal{M}(\mu_1,\mu_3)}\E_{{\Q}} [c(S_{t_1},S_{t_3})]  \\
&=\inf_{\Q \in \mathcal{M}(\mu_1,\mu_3)}\E_{{\Q}} [c(S_{t_1},S_{t_3})] 
\end{align*}
implying that $\inf_{\Q \in \mathcal{M}(\mu_1,\mu_2,\mu_3)}\E_{\Q}[c(S_{t_1},S_{t_3})]=\inf_{\Q \in \mathcal{M}(\mu_1,\mu_3)}\E_{\Q}[c(S_{t_1},S_{t_3})]$ and hence $\mu_2 \in \mathcal{I}$.
\end{itemize}
\end{proof}
}
\begin{proof}[Proof of Lemma~\ref{lem_equality_f_dualities_H}]
First, let $u_1,u_3 \in \C_b(\R)$, $u_2 \in L^1(\mu_2)$ and $\Delta_1 \in \C_b(\R), \Delta_2 \in \C_b(\R^2)$ such that 
\[
\sum_{i=1}^3 u_i(x_i)+\Delta_1(x_1)(x_{2}-x_1)+\Delta_2(x_1,x_3)(x_{3}-x_2)\leq c(x_1,x_3)~\text{ for all } (x_1,x_2,x_3)\in \R^3.
\]
Then we have by Lemma~\ref{lem_positivity_of_H}~(i) that
\begin{equation}\label{eq_ineq_u2_H}
H_{(u_i),(\Delta_i)}(x_2) \geq u_2(x_2) \text{ for all } x_2 \in \R.
\end{equation}
Moreover, note that by Lemma~\ref{lem_positivity_of_H}~(ii)
\begin{equation}\label{eq_ineq_H}
u_1(x_1) +H_{(u_i),(\Delta_i)}(x_2) + u_3(x_3)+ \Delta_1(x_1)(x_{2}-x_1)+\Delta_2(x_1,x_3)(x_{3}-x_2) \leq c(x_1,x_3) 
\end{equation}
for all $x_1,x_2,x_3\in \R^3$, and by \eqref{eq_ineq_u2_H} also that 
\begin{equation}\label{eq_H_makesitgreater}
\sum_{i=1}^3\E_{\mu_i}[u_i]\leq \E_{\mu_1}[u_1]+\E_{\mu_2}\left[H_{(u_i),(\Delta_i)}\right]+\E_{\mu_3}[u_3]
\end{equation}
which in turn  implies by using Lemma~\ref{lem_intermediate_trading} and Lemma~\ref{lem_contiinuity_H} that\footnote{Here we apply a variant of Lemma~\ref{lem_beiglboeck_dual}, where we only substitute two of the three $L^1$-integrand by functions from $\C_b(\R)$. This is possible due to the argumentation from \cite[Appendix]{beiglbock2013model} showing that for all $u\in L^1(\mu_i)$ and for all $\varepsilon>0$ there exists some $\widetilde{u} \in \C_b(\R)$ such that $\widetilde{u} \leq u$ and $\E_{\mu_i}[u]-\E_{\mu_i}[\widetilde{u}]<\varepsilon$.}
\begin{equation}
\begin{aligned}
&\sup_{u_i \in L^1(\mu_i),\Delta_i \in \mathcal{C}_b(\R^i)}\bigg\{\sum_{i=1}^3\E_{\mu_i}[u_i]~\bigg|~\sum_{i=1}^3 u_i(x_i)+\Delta_1(x_1)(x_{2}-x_1)+\Delta_2(x_1,x_3)(x_{3}-x_2)\notag \\
&\hspace{8cm}\leq c(x_1,x_3)~\text{ for all } (x_1,x_2,x_3) \in \R^3\bigg\} \notag \\
=&\sup_{u_1,u_3 \in \C_b(\R),~u_2 \in L^1(\mu_2), \atop \Delta_i \in \mathcal{C}_b(\R^i)}\bigg\{\sum_{i=1}^3\E_{\mu_i}[u_i]~\bigg|~\sum_{i=1}^3 u_i(x_i)+\Delta_1(x_1)(x_{2}-x_1)+\Delta_2(x_1,x_3)(x_{3}-x_2)\notag \\
&\hspace{8cm}\leq c(x_1,x_3)~\text{ for all } (x_1,x_2,x_3) \in \R^3\bigg\} \notag \\
\leq &\sup_{u_1,u_3 \in \C_b(\R), u_2 \equiv H_{(u_i),(\Delta_i)}, \atop \Delta_i \in \mathcal{C}_b(\R^i)}\bigg\{\sum_{i=1}^3\E_{\mu_i}[u_i]~\bigg|~\sum_{i=1}^3 u_i(x_i)+\Delta_1(x_1)(x_{2}-x_1)+\Delta_2(x_1,x_3)(x_{3}-x_2)\notag \\
&\hspace{8cm}\leq c(x_1,x_3)~\text{ for all } (x_1,x_2,x_3) \in \R^3\bigg\}. \notag 
\end{aligned}
\end{equation}
Note that, as stated in Lemma~\ref{lem_positivity_of_H}~(ii), inequality \eqref{eq_ineq_H} is always fulfilled for all $u_i \in L^1(\mu_i)$ for $i=1,3$ and $\Delta_1 \in \C_b(\R), \Delta_2 \in \C_b(\R^2)$.
This implies that
\begin{equation}\label{eq_last_equality_proof_ineq_H}
\begin{aligned}
&\sup_{u_1,u_3 \in \C_b(\R), u_2 \equiv H_{(u_i),(\Delta_i)}, \atop \Delta_i \in \mathcal{C}_b(\R^i)}\bigg\{\sum_{i=1}^3\E_{\mu_i}[u_i]~\bigg|~\sum_{i=1}^3 u_i(x_i)+\Delta_1(x_1)(x_{2}-x_1)+\Delta_2(x_1,x_3)(x_{3}-x_2) \\
&\hspace{8cm}\leq c(x_1,x_3)~\text{ for all } (x_1,x_2,x_3) \in \R^3\bigg\}  \\
= &\sup_{u_1,u_3 \in \C_b(\R), \atop \Delta_i \in \mathcal{C}_b(\R^i)}\bigg\{\E_{\mu_1}[u_1]+\E_{\mu_2}\left[H_{(u_i),(\Delta_i)}\right]+\E_{\mu_3}[u_3]~\bigg\}.  
\end{aligned}
\end{equation}
This shows one inequality of \eqref{eq_lem_equality_f_dualities_H_0}. The other inequality follows directly by using \eqref{eq_last_equality_proof_ineq_H} and by Lemma~\ref{lem_contiinuity_H}. Indeed, we have
\begin{align*}
&\sup_{u_1,u_3 \in \C_b(\R), \atop \Delta_i \in \mathcal{C}_b(\R^i)}\bigg\{\E_{\mu_1}[u_1]+\E_{\mu_2}\left[H_{(u_i),(\Delta_i)}\right]+\E_{\mu_3}[u_3]~\bigg\}.  \\
=&\sup_{u_1,u_3 \in \C_b(\R), u_2 \equiv H_{(u_i),(\Delta_i)}, \atop \Delta_i \in \mathcal{C}_b(\R^i)}\bigg\{\sum_{i=1}^3\E_{\mu_i}[u_i]~\bigg|~\sum_{i=1}^3 u_i(x_i)+\Delta_1(x_1)(x_{2}-x_1)+\Delta_2(x_1,x_3)(x_{3}-x_2) \\
&\hspace{8cm}\leq c(x_1,x_3)~\text{ for all } (x_1,x_2,x_3) \in \R^3\bigg\}  \\
\leq &\sup_{u_i \in L^1(\mu_i),\Delta_i \in \mathcal{C}_b(\R^i)}\bigg\{\sum_{i=1}^3\E_{\mu_i}[u_i]~\bigg|~\sum_{i=1}^3 u_i(x_i)+\Delta_1(x_1)(x_{2}-x_1)+\Delta_2(x_1,x_3)(x_{3}-x_2)\notag \\
&\hspace{8cm}\leq c(x_1,x_3)~\text{ for all } (x_1,x_2,x_3) \in \R^3\bigg\}.
\end{align*}
\end{proof}

\begin{proof}[Proof of Proposition~\ref{prop_degree_of_improvement}]
Note that by Lemma~\ref{lem_positivity_of_H}~(i), $H_{(u_i),(\Delta_i)}(x_2) \geq 0 $ for all $x_2 \in \R$, is equivalent to 
\begin{equation}\label{eq_H_equivalent}
u_1(x_1)+u_3(x_3) + \Delta_1(x_1)(x_2-x_1)+\Delta_2(x_1,x_2)(x_3-x_2) \leq c(x_1,x_3) \text{ for all } (x_1,x_2,x_3) \in \R^3.
\end{equation}
 We then have by Lemma~\ref{lem_beiglboeck_dual} that
 \begin{equation}\label{eq_diff_proof_eq_hequiv_1}
 \begin{aligned}
 &\inf_{\Q \in \mathcal{M}(\mu_1,\mu_2,\mu_3)}\E_\Q[c(S_{t_1},S_{t_3})]-\inf_{\Q \in \mathcal{M}(\mu_1,\mu_3)}\E_\Q[c(S_{t_1},S_{t_3})]\\
 =&  \inf_{\Q \in \mathcal{M}(\mu_1,\mu_2,\mu_3)} \E_\Q[c(S_{t_1},S_{t_3})] -  \sup_{u_1,u_3 \in \C_b(\R) \atop \Delta_1 \in \mathcal{C}_b(\R),\Delta_2 \in \mathcal{C}_b(\R^2): H_{(u_i),(\Delta_i)} \geq 0}\left(\E_{\mu_1}[u_1(S_{t_1})]+\E_{\mu_3}[u_3(S_{t_3})] \right)\\
  =&  \inf_{u_1,u_3 \in \C_b(\R) \atop \Delta_1 \in \mathcal{C}_b(\R),\Delta_2 \in \mathcal{C}_b(\R^2): H_{(u_i),(\Delta_i)} \geq 0} \inf_{\Q \in \mathcal{M}(\mu_1,\mu_2,\mu_3)}  \bigg(\E_\Q[c(S_{t_1},S_{t_3})] -  \-\left(\E_{\mu_1}[u_1(S_{t_1})]-\E_{\mu_3}[u_3(S_{t_3})] \right)\bigg).
  \end{aligned}
   \end{equation}
We now apply Lemma~\ref{lem_equality_f_dualities_H} and obtain by \eqref{eq_diff_proof_eq_hequiv_1} that
\begin{align*}
&\inf_{\Q \in \mathcal{M}(\mu_1,\mu_2,\mu_3)}\E_\Q[c(S_{t_1},S_{t_3})]-\inf_{\Q \in \mathcal{M}(\mu_1,\mu_3)}\E_\Q[c(S_{t_1},S_{t_3})]\\
 =&\inf_{u_1,u_3 \in \C_b(\R) \atop \Delta_1 \in \mathcal{C}_b(\R),\Delta_2 \in \mathcal{C}_b(\R^2): H_{(u_i),(\Delta_i)} \geq 0}  \sup_{v_1,v_2 \in \C_b(\R)\atop \widetilde{\Delta}_1 \in \mathcal{C}_b(\R),\widetilde{\Delta}_2 \in \mathcal{C}_b(\R^2)} \bigg( \E_{\mu_1}[v_1(S_{t_1})]+\E_{\mu_2}\left[ H_{(v_i),(\widetilde{\Delta}_i)}\right]+\E_{\mu_3}[v_3(S_{t_3})] \\
 &\hspace{10cm} -\E_{\mu_1}[u_1(S_{t_1})]-\E_{\mu_3}[u_3(S_{t_3})]\bigg)\\
 =&\inf_{u_1,u_3 \in \C_b(\R) \atop \Delta_1 \in \mathcal{C}_b(\R),\Delta_2 \in \mathcal{C}_b(\R^2): H_{(u_i),(\Delta_i)} \geq 0}  \sup_{v_1,v_2 \in \C_b(\R)\atop \widetilde{\Delta}_1 \in \mathcal{C}_b(\R),\widetilde{\Delta}_2 \in \mathcal{C}_b(\R^2)} \bigg(\E_{\mu_2}\left[ H_{(u_i-v_i),(\Delta_i-\widetilde{\Delta}_i)}\right]-\E_{\mu_1}[v_1(S_{t_1})]-\E_{\mu_3}[v_3(S_{t_3})]\bigg).
 \end{align*}
\end{proof}

\begin{proof}[Proof of Remark~\ref{rem_nmarginals}]
This follows analogue to the proof of Proposition~\ref{prop_degree_of_improvement}.
\end{proof}

\begin{proof}[Proof of Proposition~\ref{prop_mot_finite_options}]
We show the assertion only for $\underline{P}_{1,2,3}(c)$, the case $\underline{P}_{1,3}(c)$ follows completely analogue. To this end, note that the inequality $\underline{P}_{1,2,3}(c) \leq \inf_{\Q\in \mathcal{M}(\mu_1^*,\mu_2^*,\mu_3^*)}\E_\Q[c(S_{t_1},S_{t_3})]$ follows by definition of the super-replication functional and we only need to show the converse inequality. 
Since all joint distributions of the marginals $\mu_i^*$, $i=1,2,3$ are supported on the set $\Xi:= \left\{(K_{1,j_1},K_{2,j_2},K_{3,j_3})~\middle|~j_i \in \{1,\dots,m_i\} \text{ for } i =1,2,3\right\}$, we have
\begin{align*}
\inf_{\Q\in \mathcal{M}(\mu_1^*,\mu_2^*,\mu_3^*)}\E_\Q[c(S_{t_1},S_{t_3})]&=\inf_{\Q\in \mathcal{M}(\mu_1^*,\mu_2^*,\mu_3^*)} \left\{\E_\Q[c(S_{t_1},S_{t_3})]~\middle|~~\Q(\Xi)=1\right\}.
\end{align*}
With an application of  Lemma~\ref{lem_xi_dual} it follows that 
\begin{align*}
&\inf_{\Q\in \mathcal{M}(\mu_1^*,\mu_2^*,\mu_3^*)} \left\{\E_\Q[c(S_{t_1},S_{t_3})]~\middle|~~\Q(\Xi)=1\right\}\\
= &\sup_{u_i \in \Clin(\R), \atop \Delta_i \in \mathcal{C}_b(\R^i)}\bigg\{\sum_{i=1}^3\E_{\mu_i^*}[u_i]~\bigg|~\sum_{i=1}^3u_i(x_i)+\sum_{i=1}^{2}\Delta_i(x_1,x_i)(x_{i+1}-x_i)\notag \\
&\hspace{6.5cm}\leq c(x_1,x_3)~\text{ for all } (x_1,x_2,x_3) \in \Xi \bigg\}.
\end{align*}
Lemma~\ref{lem_finite_options} now implies that we can switch to integrands from the class $${\mathcal{S}}_i :=\left\{u:\R \rightarrow \R~\middle|~u(x)=d+\sum_{j=0}^{m_i}\lambda_j (x-K_{i,j})^++\Delta_0(x-S_0) \text{ for some } d, \Delta_0 , \lambda_j \in \R \right\}$$ for $i=1,2,3$.
 This means we have
\begin{align*}
&\sup_{u_i \in \Clin(\R), \atop \Delta_i \in \mathcal{C}_b(\R^i)}\bigg\{\sum_{i=1}^3\E_{\mu_i^*}[u_i]~\bigg|~\sum_{i=1}^3u_i(x_i)+\sum_{i=1}^{2}\Delta_i(x_1,x_i)(x_{i+1}-x_i)\notag \\
&\hspace{6.5cm}\leq c(x_1,x_3)~\text{ for all } (x_1,x_2,x_3) \in \Xi \bigg\} \\
\leq  &\sup_{u_i \in {\mathcal{S}}_i,\atop \Delta_i \in \mathcal{C}_b(\R^i)}\bigg\{\sum_{i=1}^3\E_{\mu_i^*}[u_i]~\bigg|~\sum_{i=1}^3u_i(x_i)+\sum_{i=1}^{2}\Delta_i(x_1,x_i)(x_{i+1}-x_i)\notag \\
&\hspace{6.5cm}\leq c(x_1,x_3)~\text{ for all } (x_1,x_2,x_3) \in \Xi \bigg\}.
\end{align*}
 By taking into account that $\E_{\mu_i}[S_{t_i}]=S_0$ and that $\E_{\mu_i}[(S_{t_i}-K_{i,j})^+]=\Pi_{i,j}$ for all $i=1,2,3$ $j=0,\dots,m_i$ we obtain
\begin{align*}
&\sup_{u_i \in {\mathcal{S}}_i, \atop \Delta_i \in \mathcal{C}_b(\R^i)}\bigg\{\sum_{i=1}^3\E_{\mu_i^*}[u_i]~\bigg|~\sum_{i=1}^3u_i(x_i)+\sum_{i=1}^{2}\Delta_i(x_1,x_i)(x_{i+1}-x_i)\leq c(x_1,x_3)~\text{ for all } (x_1,x_2,x_3) \in  \Xi \bigg\} \\
&= \sup_{d_i,\lambda_{i,j}, \Delta_0^i \in \R, \atop \Delta_i \in \mathcal{C}_b(\R^i)}\bigg\{\sum_{i=1}^3\left(d_i+\sum_{j=0}^{m_i}\lambda_{i,j}\Pi_{i,j}\right)~\bigg|~\sum_{i=1}^3\left(d_i+\sum_{j=0}^{m_i}\lambda_{i,j}(x_i-K_{i,j})^++\Delta_0^i(x_i-S_0)\right)\notag \\
&\hspace{5.5cm}+\sum_{i=1}^{2}\Delta_i(x_1,x_i)(x_{i+1}-x_i)\leq c(x_1,x_3)~\text{ for all } (x_1,x_2,x_3) \in  \Xi\bigg\}\\
&= \sup_{d,\lambda_{i,j}, \Delta_0 \in \R,\atop \Delta_i \in \mathcal{C}_b(\R^i)}\bigg\{d+ \sum_{i=1}^3\sum_{j=0}^{m_i}\lambda_{i,j}\Pi_{i,j}~\bigg|~d+\sum_{i=1}^3\sum_{j=0}^{m_i}\lambda_{i,j}(x_i-K_{i,j})^++\Delta_0(x_1-S_0) \\
&\hspace{5.5cm}+\sum_{i=1}^{2}\Delta_i(x_1,x_i)(x_{i+1}-x_i) \leq c(x_1,x_3)~\text{ for all } (x_1,x_2,x_3) \in \Xi\bigg\}\\
&=\underline{P}_{1,2,3}(c).
\end{align*}
\end{proof}

\begin{proof}[Proof of Corollary~\ref{cor_improvement_call}]
The first equality follows directly by the definition of $\underline{P}_{1,3}(c)$ and $\underline{P}_{1,2,3}(c)$ and by Lemma~\ref{lem_financial_intermediate}. The second equality follows by Proposition~\ref{prop_mot_finite_options} and Proposition~\ref{prop_degree_of_improvement}.
\end{proof}

{
\begin{proof}[Proof of Proposition~\ref{prop_intermediate_Tu_Td}] First assume that there is no improvement, i.e., the equality $$
\inf_{\Q \in \mathcal{M}(\mu_1,\mu_2,\mu_3)}\E_\Q[c(S_{t_1},S_{t_3})]=\inf_{\Q \in \mathcal{M}(\mu_1,\mu_3)}\E_\Q[c(S_{t_1},S_{t_3})]
$$ holds. Then, according to Proposition~\ref{prop_intermediate_projection}~(vi), there exists some probability measure $\Q \in \mathcal{M}(\mu_1,\mu_2,\mu_3)$ with $\pi(\Q)=\Q^* \in \mathcal{Q}_c^*(\mu_1,\mu_3)$. In particular, under $\Q$, given knowledge of $S_{t_1}$ and $S_{t_2}$, the law of $S_{t_3}$ is supported on $\{T_d(S_{t_1}),T_u(S_{t_1})\}$, and we have
\[
p(S_{t_1},S_{t_2}):= \Q\left(S_{t_3}=T_u(S_{t_1})~|~S_{t_1},S_{t_2}\right)=1-\Q\left(S_{t_3}=T_d(S_{t_1})~|~S_{t_1},S_{t_2}\right)
\]
The martingale property implies
\[
S_{t_2} = \E_\Q[S_{t_3}~|~S_{t_1},S_{t_2}]= T_u(S_{t_1})p(S_{t_1},S_{t_2})+T_d(S_{t_1})\left(1-p(S_{t_1},S_{t_2})\right)
\]
Hence, we have $\Q$-almost surely that
\[
p(S_{t_1},S_{t_2})= \frac{S_{t_2}-T_{d}(S_{t_1})}{T_{u}(S_{t_1})-T_{d}(S_{t_1})} \one_{\{T_u(S_{t_1})>T_{d}(S_{t_1})\}}
\]
Since $p(S_{t_1},S_{t_2}) \in [0,1]$ we have $\Q$-almost surely
\[
1=\Q\bigg(T_d(S_{t_1}) \leq S_{t_2} \leq T_u(S_{t_1})~\bigg|~T_u(S_{t_1})>T_{d}(S_{t_1})\bigg)\text{ and } 1=\Q\bigg(S_{t_2}=S_{t_1}~\bigg|~T_u(S_{t_1})=T_{d}(S_{t_1})\bigg).
\]
Hence $1=\Q\left(T_d(S_{t_1}) \leq S_{t_2} \leq T_u(S_{t_1})\right)$ and projecting $\Q$ on the first two marginals shows the existence of some $\widetilde{\Q} \in \mathcal{M}(\mu_1,\mu_2)$ with $1=\widetilde{\Q}\left(T_d(S_{t_1}) \leq S_{t_2} \leq T_u(S_{t_1})\right)$.

To show the reverse implication, consider some martingale measure $\widetilde{\Q} \in \mathcal{M}(\mu_1,\mu_2)$ with \begin{equation}\label{equ_assumption_mass_1}
1=\widetilde{\Q}\left(T_d(S_{t_1}) \leq S_{t_2} \leq T_u(S_{t_1})\right).
\end{equation}
 By the disintegration theorem (\cite[Theorem 1]{chang1997conditioning} or \cite[Theorem 5.3.1]{ambrosio2005gradient}), there exists some probability kernel $\Q_1$ such that 
\[
\widetilde{\Q}(\D x_1, \D x_2) = \mu_1(\D x_1) \Q_1(x_1; \D x_2)
\]
Define further the probability kernel $\Q_{1,2}$ by 
\[
\Q_{1,2}(x_1,x_2;\D x_3):=\left(\widetilde{q}(x_1,x_2)\delta_{T_u(x_1)}(x_3)+(1-\widetilde{q}(x_1,x_2))\delta_{T_d(x_1)}(x_3)\right)\D x_3
\]
where $\widetilde{q}(x_1,x_2) = \frac{x_2- T_d(x_1)}{T_u(x_1)-T_d(x_1)} \one_{\{T_u(x_1)>T_d(x_1)\}}$. Then
\[
\Q(\D x_1,\D x_2, \D x_3) := \mu_1( \D  x_1) \Q_1( x_1; \D x_2) \Q_{1,2}(x_1, x_2; \D x_3)
\]
defines a probability measure $\Q \in \mathcal{M}(\mu_1,\mu_2,\mu_3)$ with $\pi(\Q) =\Q^* \in \mathcal{Q}_c^*(\mu_1,\mu_3)$. Indeed, we have 
\begin{equation}\label{eq_projection_qstar_q}
\widetilde{\Q}= \Q \circ (S_{t_1},S_{t_2})^{-1}
\end{equation}
and hence $\Q$ possesses the fist two marginals $\mu_1$, $\mu_2$ and it holds $\Q$-almost surely
\begin{equation}\label{eq_martingale_last_proof_1}
\E_\Q[S_{t_2}~|~S_{t_1}]=\E_{\widetilde{\Q}}[S_{t_2}~|~S_{t_1}]=S_{t_1}
\end{equation}
Moreover, if $T_u(S_{t_1})=T_d(S_{t_1})$, then by definition $\E_\Q[S_{t_3}~|~S_{t_2},S_{t_1}]=T_d(S_{t_1})=S_{t_2}$ $\Q$-almost surely, where the last equality follows from \eqref{equ_assumption_mass_1}.
If $T_{u}(S_{t_1})>T_{u}(S_{t_1})$, then $\Q$-almost surely
\begin{equation}\label{eq_martingale_last_proof_2}
\begin{aligned}
\E_\Q[S_{t_3}~|~S_{t_2},S_{t_1}]&=\tilde{q}(S_{t_1},S_{t_2})T_u(S_{t_1})+(1-\tilde{q}(S_{t_1},S_{t_2}))T_d(S_{t_1})\\
&=\frac{T_u(S_{t_1})S_{t_2}-T_u(S_{t_1})T_d(S_{t_1})+T_u(S_{t_1})T_d(S_{t_1})-S_{t_2}T_d(S_{t_1})}{T_u(S_{t_1})-T_d(S_{t_1})}=S_{t_2}.
\end{aligned}
\end{equation}
Next, note that by \eqref{equ_assumption_mass_1} and \eqref{eq_martingale_last_proof_1} we have 
\begin{align*}
\E_{\Q}[\widetilde{q}(S_{t_1},S_{t_2})~|~S_{t_1}]&=\E_{\Q}\left[\frac{S_{t_2}- T_d(S_{t_1})}{T_u(S_{t_1})-T_d(S_{t_1})} \one_{\{T_u(S_{t_1})>T_d(S_{t_1})\}}~\middle|~S_{t_1}\right] \\
&=\frac{S_{t_1}- T_d(S_{t_1})}{T_u(S_{t_1})-T_d(S_{t_1})} \one_{\{T_u(S_{t_1})>T_d(S_{t_1})\}} = q(S_{t_1})
\end{align*}
$\Q$-almost surely for $q$ being defined in \eqref{eq_definition_prob_q}. This implies  for all Borel-measurable sets $A,B \subseteq \R$ that
\begin{align*}
\Q(S_{t_3}\in B)&= \E_{\Q}\left[\E_{\Q}[\widetilde{q}(S_{t_1},S_{t_2})\one_{\{T_u(S_{t_1}) \in B\}}+(1-\widetilde{q}(S_{t_1},S_{t_2}))\one_{\{T_d(S_{t_1}) \in B\}}~|~S_{t_1}]\right]\\
&= \E_{\mu_1}\left[q(S_{t_1})\one_{\{T_u(S_{t_1}) \in B\}}+(1-q(S_{t_1})\one_{\{T_d(S_{t_1}) \in B\}}\right] \\
&= \E_{\Q^*}[\one_{\{S_{t_3} \in B\}}]= \Q^* \left(S_{t_3} \in B\right) = \mu_3\left(B\right)
\end{align*}
and 
\begin{align*}
&\Q(S_{t_1} \in A,S_{t_2} \in \R, S_{t_3}\in B)\\
&=  \E_{\Q}\left[\E_{\Q}[\widetilde{q}(S_{t_1},S_{t_2})\one_{\{S_{t_1} \in A, T_u(S_{t_1}) \in B\}}+(1-\widetilde{q}(S_{t_1},S_{t_2}))\one_{\{S_{t_1} \in A,T_d(S_{t_1}) \in B\}}~|~S_{t_1}]\right]\\
&= \E_{\mu_1}\left[q(S_{t_1})\one_{\{S_{t_1} \in A,T_u(S_{t_1}) \in B\}}+(1-q(S_{t_1})\one_{\{S_{t_1} \in A,T_d(S_{t_1}) \in B\}}\right] \\
&= \E_{\Q^*}[\one_{\{S_{t_1} \in A,S_{t_3} \in B\}}] = \Q^*\left(S_{t_1} \in A,S_{t_3} \in B\right)
\end{align*}
 Hence, we have shown $\Q \in \mathcal{M}(\mu_1,\mu_2,\mu_3)$ with $\pi(\Q) =\Q^* \in \mathcal{Q}_c^*(\mu_1,\mu_3)$, and according to Proposition~\ref{prop_intermediate_projection}~(vi) the price bounds are not improved.
\end{proof}
}

\section*{Acknowledgements}
I thank David Hobson, Eva Lütkebohmert and Stephan Eckstein for valuable comments and fruitful discussions on the topic. I am also very grateful for helpful advises from two anonymous referees that helped to improve the manuscript significantly.

\section*{Data Availability}
The datasets used in Example~\ref{exa_real} are provided from \emph{Yahoo Finance} but restrictions apply to the availability of these data, since only current option prices can be downloaded from \emph{Yahoo Finance}, and so the used data is not publicly available. The data is however available from the author
upon reasonable request and permission of \emph{Yahoo Inc}.

%
%


\begin{thebibliography}{10}

\bibitem{acciaio2016model}
Beatrice Acciaio, Mathias Beiglb{\"o}ck, Friedrich Penkner, and Walter
  Schachermayer.
\newblock A model-free version of the fundamental theorem of asset pricing and
  the super-replication theorem.
\newblock {\em Mathematical Finance}, 26(2):233--251, 2016.

\bibitem{acciaio2017space}
Beatrice Acciaio, Martin Larsson, and Walter Schachermayer.
\newblock The space of outcomes of semi-static trading strategies need not be
  closed.
\newblock {\em Finance and Stochastics}, 21(3):741--751, 2017.

\bibitem{aksamit2020robust}
Anna Aksamit, Zhaoxu Hou, and Jan Ob{\l}{\'o}j.
\newblock Robust framework for quantifying the value of information in pricing
  and hedging.
\newblock {\em SIAM Journal on Financial Mathematics}, 11(1):27--59, 2020.

\bibitem{alfonsi2019sampling}
Aur{\'e}lien Alfonsi, Jacopo Corbetta, and Benjamin Jourdain.
\newblock Sampling of one-dimensional probability measures in the convex order
  and computation of robust option price bounds.
\newblock {\em International Journal of Theoretical and Applied Finance},
  22(03):1950002, 2019.

\bibitem{ambrosio2005gradient}
Luigi Ambrosio, Nicola Gigli, and Giuseppe Savar{\'e}.
\newblock {\em Gradient flows: in metric spaces and in the space of probability
  measures}.
\newblock Springer Science \& Business Media, 2005.

\bibitem{ansari2022improved}
Jonathan Ansari, Eva L{\"u}tkebohmert, Ariel Neufeld, and Julian Sester.
\newblock Improved robust price bounds for multi-asset derivatives under
  market-implied dependence information.
\newblock {\em arXiv preprint arXiv:2204.01071}, 2022.

\bibitem{backhoff2022stability}
Julio Backhoff-Veraguas and Gudmund Pammer.
\newblock Stability of martingale optimal transport and weak optimal transport.
\newblock {\em The Annals of Applied Probability}, 32(1):721--752, 2022.

\bibitem{bartl2019robust}
Daniel Bartl, Patrick Cheridito, and Michael Kupper.
\newblock Robust expected utility maximization with medial limits.
\newblock {\em Journal of Mathematical Analysis and Applications},
  471(1-2):752--775, 2019.

\bibitem{bartl2020pathwise}
Daniel Bartl, Michael Kupper, and Ariel Neufeld.
\newblock Pathwise superhedging on prediction sets.
\newblock {\em Finance and Stochastics}, 24(1):215--248, 2020.

\bibitem{bauerle2019consistent}
Nicole B{\"a}uerle and Daniel Schmithals.
\newblock Consistent upper price bounds for exotic options.
\newblock {\em International Journal of Theoretical and Applied Finance},
  24(02):2150011, 2021.

\bibitem{bayraktar2022supermartingale}
Erhan Bayraktar, Shuoqing Deng, and Dominykas Norgilas.
\newblock Supermartingale shadow couplings: the decreasing case.
\newblock {\em arXiv preprint arXiv:2207.11732}, 2022.

\bibitem{beiglbock2013model}
Mathias Beiglb{\"o}ck, Pierre Henry-Labord{\`e}re, and Friedrich Penkner.
\newblock Model-independent bounds for option prices—a mass transport
  approach.
\newblock {\em Finance and Stochastics}, 17(3):477--501, 2013.

\bibitem{beiglbock2022potential}
Mathias Beiglb{\"o}ck, David Hobson, and Dominykas Norgilas.
\newblock The potential of the shadow measure.
\newblock {\em Electronic Communications in Probability}, 27:1--12, 2022.

\bibitem{beiglboeck2016problem}
Mathias Beiglb{\"o}ck and Nicolas Juillet.
\newblock On a problem of optimal transport under marginal martingale
  constraints.
\newblock {\em The Annals of Probability}, 44(1):42--106, 2016.

\bibitem{beiglbock2016problem}
Mathias Beiglb{\"o}ck and Nicolas Juillet.
\newblock On a problem of optimal transport under marginal martingale
  constraints.
\newblock {\em The Annals of Probability}, 44(1):42--106, 2016.

\bibitem{beiglbock2021shadow}
Mathias Beiglb{\"o}ck and Nicolas Juillet.
\newblock Shadow couplings.
\newblock {\em Transactions of the American Mathematical Society},
  374(7):4973--5002, 2021.

\bibitem{beiglbock2019dual}
Mathias Beiglb{\"o}ck, Tongseok Lim, and Jan Ob{\l}{\'o}j.
\newblock Dual attainment for the martingale transport problem.
\newblock {\em Bernoulli}, 25(3):1640--1658, 2019.

\bibitem{beiglbock2017complete}
Mathias Beiglb{\"o}ck, Marcel Nutz, and Nizar Touzi.
\newblock Complete duality for martingale optimal transport on the line.
\newblock {\em The Annals of Probability}, 45(5):3038--3074, 2017.

\bibitem{breeden1978prices}
Douglas~T. Breeden and Robert~H. Litzenberger.
\newblock Prices of state-contingent claims implicit in option prices.
\newblock {\em Journal of Business}, 51(4):621--651, 1978.

\bibitem{burzoni2016arbitrage}
Matteo Burzoni.
\newblock Arbitrage and hedging in model-independent markets with frictions.
\newblock {\em SIAM Journal on Financial Mathematics}, 7(1):812--844, 2016.

\bibitem{burzoni2017model}
Matteo Burzoni, Marco Frittelli, and Marco Maggis.
\newblock Model-free superhedging duality.
\newblock {\em The Annals of Applied Probability}, 27(3):1452--1477, 2017.

\bibitem{chang1997conditioning}
Joseph~T Chang and David Pollard.
\newblock Conditioning as disintegration.
\newblock {\em Statistica Neerlandica}, 51(3):287--317, 1997.

\bibitem{cheridito2021martingale}
Patrick Cheridito, Matti Kiiski, David~J Pr{\"o}mel, and H~Mete Soner.
\newblock Martingale optimal transport duality.
\newblock {\em Mathematische Annalen}, 379(3):1685--1712, 2021.

\bibitem{cheridito2017duality}
Patrick Cheridito, Michael Kupper, and Ludovic Tangpi.
\newblock Duality formulas for robust pricing and hedging in discrete time.
\newblock {\em SIAM Journal on Financial Mathematics}, 8(1):738--765, 2017.

\bibitem{cohen2020detecting}
Samuel~N Cohen, Christoph Reisinger, and Sheng Wang.
\newblock Detecting and repairing arbitrage in traded option prices.
\newblock {\em Applied Mathematical Finance}, 27(5):345--373, 2020.

\bibitem{cousot2007conditions}
Laurent Cousot.
\newblock Conditions on option prices for absence of arbitrage and exact
  calibration.
\newblock {\em Journal of Banking \& Finance}, 31(11):3377--3397, 2007.

\bibitem{cox2019structure}
Alexander~MG Cox and Matija Vidmar.
\newblock The structure of non-linear martingale optimal transport problems.
\newblock {\em arXiv preprint arXiv:1903.06606}, 2019.

\bibitem{davis2007range}
Mark~HA Davis and David~G Hobson.
\newblock The range of traded option prices.
\newblock {\em Mathematical Finance}, 17(1):1--14, 2007.

\bibitem{de2020bounds}
Luca de~Gennaro~Aquino and Carole Bernard.
\newblock Bounds on multi-asset derivatives via neural networks.
\newblock {\em International Journal of Theoretical and Applied Finance},
  23(08):2050050, 2020.

\bibitem{doldi2020entropy}
Alessandro Doldi and Marco Frittelli.
\newblock Entropy martingale optimal transport and nonlinear pricing-hedging
  duality.
\newblock {\em arXiv preprint arXiv:2005.12572}, 2020.

\bibitem{dolinsky2014robust}
Yan Dolinsky and H~Mete Soner.
\newblock Robust hedging with proportional transaction costs.
\newblock {\em Finance and Stochastics}, 18(2):327--347, 2014.

\bibitem{eckstein2021robust}
Stephan Eckstein, Gaoyue Guo, Tongseok Lim, and Jan Ob{\l}{\'o}j.
\newblock Robust pricing and hedging of options on multiple assets and its
  numerics.
\newblock {\em SIAM Journal on Financial Mathematics}, 12(1):158--188, 2021.

\bibitem{eckstein2021martingale}
Stephan Eckstein and Michael Kupper.
\newblock Martingale transport with homogeneous stock movements.
\newblock {\em Quantitative Finance}, 21(2):271--280, 2021.

\bibitem{ekren2018constrained}
Ibrahim Ekren and H~Mete Soner.
\newblock Constrained optimal transport.
\newblock {\em Archive for Rational Mechanics and Analysis}, 227(3):929--965,
  2018.

\bibitem{follmer2016stochastic}
Hans F{\"o}llmer and Alexander Schied.
\newblock Stochastic finance.
\newblock In {\em Stochastic Finance}. de Gruyter, 2016.

\bibitem{ghoussoub2019structure}
Nassif Ghoussoub, Young-Heon Kim, and Tongseok Lim.
\newblock Structure of optimal martingale transport plans in general
  dimensions.
\newblock {\em The Annals of Probability}, 47(1):109--164, 2019.

\bibitem{guo2019computational}
Gaoyue Guo and Jan Ob{\l}{\'o}j.
\newblock Computational methods for martingale optimal transport problems.
\newblock {\em The Annals of Applied Probability}, 29(6):3311--3347, 2019.

\bibitem{henry2013automated}
Pierre Henry-Labordere.
\newblock Automated option pricing: Numerical methods.
\newblock {\em International Journal of Theoretical and Applied Finance},
  16(08):1350042, 2013.

\bibitem{henry2017modelfree}
Pierre Henry-Labord{\`e}re.
\newblock {\em Model-free hedging, A martingale optimal transport viewpoint}.
\newblock Financial Mathematics Series. Chapman and Hall / CRC, 2017.

\bibitem{henry2016explicit_full}
Pierre Henry-Labord{\`e}re, Xiaolu Tan, and Nizar Touzi.
\newblock An explicit martingale version of the one-dimensional {Brenier}’s
  theorem with full marginals constraint.
\newblock {\em Stochastic Processes and their Applications}, 126(9):2800--2834,
  2016.

\bibitem{henry2016explicit}
Pierre Henry-Labord{\`e}re and Nizar Touzi.
\newblock An explicit martingale version of the one-dimensional {Brenier}
  theorem.
\newblock {\em Finance and Stochastics}, 20(3):635--668, 2016.

\bibitem{hobson2015robust}
David Hobson and Martin Klimmek.
\newblock Robust price bounds for the forward starting straddle.
\newblock {\em Finance and Stochastics}, 19(1):189--214, 2015.

\bibitem{hobson2012robust}
David Hobson and Anthony Neuberger.
\newblock Robust bounds for forward start options.
\newblock {\em Mathematical Finance}, 22(1):31--56, 2012.

\bibitem{hobson2019robust}
David Hobson and Dominykas Norgilas.
\newblock Robust bounds for the american put.
\newblock {\em Finance and Stochastics}, 23(2):359--395, 2019.

\bibitem{hobson2019left}
David~G Hobson and Dominykas Norgilas.
\newblock The left-curtain martingale coupling in the presence of atoms.
\newblock {\em The Annals of Applied Probability}, 29(3):1904--1928, 2019.

\bibitem{hou2018robust}
Zhaoxu Hou and Jan Ob{\l}{\'o}j.
\newblock Robust pricing--hedging dualities in continuous time.
\newblock {\em Finance and Stochastics}, 22(3):511--567, 2018.

\bibitem{hull2003options}
John~C Hull.
\newblock {\em Options futures and other derivatives}.
\newblock Pearson Education India, 2003.

\bibitem{juillet2016stability}
Nicolas Juillet.
\newblock Stability of the shadow projection and the left-curtain coupling.
\newblock In {\em Annales de l'Institut Henri Poincar{\'e}, Probabilit{\'e}s et
  Statistiques}, volume~52, pages 1823--1843. Institut Henri Poincar{\'e},
  2016.

\bibitem{kellerer1972markov}
Hans~G Kellerer.
\newblock Markov-komposition und eine anwendung auf martingale.
\newblock {\em Mathematische Annalen}, 198(3):99--122, 1972.

\bibitem{lutkebohmert2019tightening}
Eva L{\"u}tkebohmert and Julian Sester.
\newblock Tightening robust price bounds for exotic derivatives.
\newblock {\em Quantitative Finance}, 19(11):1797--1815, 2019.

\bibitem{lux2017improved}
Thibaut Lux and Antonis Papapantoleon.
\newblock Improved fr{\'e}chet--hoeffding bounds on $ d $-copulas and
  applications in model-free finance.
\newblock {\em The Annals of Applied Probability}, 27(6):3633--3671, 2017.

\bibitem{neufeld2022model}
Ariel Neufeld, Antonis Papapantoleon, and Qikun Xiang.
\newblock Model-free bounds for multi-asset options using option-implied
  information and their exact computation.
\newblock {\em Management Science}, 2022.

\bibitem{neufeld2021model}
Ariel Neufeld and Julian Sester.
\newblock Model-free price bounds under dynamic option trading.
\newblock {\em SIAM Journal on Financial Mathematics}, 12(4):1307--1339, 2021.

\bibitem{neufeld2021stability}
Ariel Neufeld and Julian Sester.
\newblock On the stability of the martingale optimal transport problem: A
  set-valued map approach.
\newblock {\em Statistics \& Probability Letters}, 176:109131, 2021.

\bibitem{nutz2017multiperiod}
Marcel Nutz, Florian Stebegg, and Xiaowei Tan.
\newblock Multiperiod martingale transport.
\newblock {\em Stochastic Processes and their Applications}, 2019.

\bibitem{nutz2022limits}
Marcel Nutz, Johannes Wiesel, and Long Zhao.
\newblock Limits of semistatic trading strategies.
\newblock {\em arXiv preprint arXiv:2204.12251}, 2022.

\bibitem{rothschild1978increasing}
Michael Rothschild and Joseph~E Stiglitz.
\newblock Increasing risk: I. a definition.
\newblock In {\em Uncertainty in Economics}, pages 99--121. Elsevier, 1978.

\bibitem{sester2020robust}
Julian Sester.
\newblock Robust bounds for derivative prices in markovian models.
\newblock {\em International Journal of Theoretical and Applied Finance},
  23(03):2050015, 2020.

\bibitem{shaked2007stochastic}
Moshe Shaked and J~George Shanthikumar.
\newblock {\em Stochastic orders}.
\newblock Springer, 2007.

\bibitem{strassen1965existence}
Volker Strassen.
\newblock The existence of probability measures with given marginals.
\newblock {\em The Annals of Mathematical Statistics}, 36(2):423--439, 1965.

\bibitem{tankov2011improved}
Peter Tankov.
\newblock Improved {F}r{\'e}chet bounds and model-free pricing of multi-asset
  options.
\newblock {\em Journal of Applied Probability}, 48(2):389--403, 2011.

\bibitem{villani2009optimal}
C{\'e}dric Villani et~al.
\newblock {\em Optimal transport: old and new}, volume 338.
\newblock Springer, 2009.

\bibitem{wiesel2019continuity}
Johannes Wiesel.
\newblock Continuity of the martingale optimal transport problem on the real
  line.
\newblock {\em arXiv preprint arXiv:1905.04574}, 2019.

\bibitem{zaev2015monge}
Danila~A Zaev.
\newblock On the {M}onge--{K}antorovich problem with additional linear
  constraints.
\newblock {\em Mathematical Notes}, 98(5):725--741, 2015.

\end{thebibliography}


\small 
 \newcommand{\noop}[1]{forthcoming}

\bibliographystyle{plain}

\end{document}